\documentclass[11pt]{article}

\usepackage{graphicx}
\usepackage{latexsym}
\usepackage[usenames,dvipsnames]{color}
\usepackage{amssymb,amsmath}
\usepackage{lscape}

\normalbaselines
\catcode `\@=11
\@addtoreset{equation}{section}

\def\section{\@startsection {section}{1}{\z@}{-3.5ex plus -1ex minus -.2ex}{2.3ex plus .2ex}{\large\bf}} 
\def\subsection{ \@startsection{subsection}{2}{\z@}{-3.25ex plus -1ex minus -.2ex}{1.5ex plus .2ex}{\normalsize \bf}}
\def\subsubsection{\@startsection{subsubsection}{3}{\z@}{-3.25ex plus -1ex minus -.2ex}{1.5ex plus .2ex}{\normalsize\sl}} 
\catcode`\@=12

\newtheorem{theorem}{Theorem}[section]
\newtheorem{def.}[theorem]{Definition}
\newtheorem{prop.}[theorem]{Proposition}
\newtheorem{lem.}[theorem]{Lemma}
\newtheorem{cor.}[theorem]{Corollary}
\newtheorem{conj.}[theorem]{Conjecture}
\newtheorem{Bsp.}{Example}[section]
\newtheorem{rem.}{Remarks}[theorem]
\newtheorem{rem}{Remark}[theorem]

\newenvironment{proof}{\noindent \bf Proof: \rm}{$ \hspace{\stretch{1}} \Box $}

 \newcommand{\bea}{\begin{eqnarray}}
\newcommand{\ena}{\end{eqnarray}}

\newcommand{\beano}{\begin{eqnarray*}}
\newcommand{\enano}{\end{eqnarray*}}
 
\newcommand{\bei}{\begin{itemize}}
\newcommand{\eni}{\end{itemize}}

\newcommand{\be}{\begin{equation}}
\newcommand{\en}{\end{equation}}

\newcommand{\belem}{\begin{lem.}}
\newcommand{\enlem}{\end{lem.}}

\newcommand{\berem}{\begin{rem} \rm}
\newcommand{\enrem}{\end{rem}}

\newcommand{\norm}[2]{
\left\| #2 \right\|_{#1}
}
\newcommand{\Hil}[0]{\mathcal{H} }

\renewcommand{\ge}{\geqslant}
\renewcommand{\le}{\leqslant}
\renewcommand{\leq}{\leqslant}
\renewcommand{\geq}{\geqslant}

\newcommand{\nn}{\nonumber}
 \newcommand{\noi}{\noindent}

\newcommand{\ha}{^{\rm\textstyle *}}
\newcommand{\haa}{^{\rm\textstyle **}}
 \newcommand{\ov}{\overline}

\newcommand{\RN}{\mathbb{R}}
\newcommand{\ZN}{\mathbb{Z}}
\newcommand{\CN}{{\mathbb C}}

\newcommand{\dom}{{\sf Dom}}
\newcommand{\ran}{{\sf Ran}}

\def\B{{\mathcal B}}

\def\F{{\mathcal F}}

\def\H{{\mathcal H}}
\def\HH{\H}
\def\h{{\mathfrak H}}
\def\s{{\mathfrak S}}

\newcommand{\NN}[0]{\mathbb{N}}
\newcommand{\ud}{\,\mathrm{d}}
\newcommand{\up}{\raisebox{0.7mm}{$\upharpoonright$}}
 
\newcommand{\ip}[2]{\langle {#2},{#1}\rangle}

\def\<{\langle}
\def\>{\rangle}
\def\xxl{\color{Brown}}

\begin{document}

\begin{flushleft}
{\Large \sc Frames, semi-frames, and Hilbert scales} \vspace*{7mm}

{\large\sf   J-P. Antoine $\!^{\rm a}$\footnote{{\it E-mail address}: jean-pierre.antoine@uclouvain.be} and  
P. Balazs $\!^{\rm b}$\footnote{{\it E-mail address}: peter.balazs@oeaw.ac.at} 
}
\\[3mm]
$^{\rm a}$ \emph{\small Institut de Recherche en Math\'ematique et  Physique, Universit\'e catholique de Louvain \\
\hspace*{3mm}B-1348   Louvain-la-Neuve, Belgium}
\\[1mm]
$^{\rm b}$ \emph{\small Acoustics Research Institute, Austrian Academy of Sciences\\
\hspace*{3mm}A-1040 Vienna, Austria }
\end{flushleft}

 \begin{abstract} 
Given a total sequence in a Hilbert space, 
we speak of an upper (resp. lower) semi-frame if  only the upper (resp. lower)  frame bound is valid. 
Equivalently, for an  upper  semi-frame, the frame operator is bounded, but has an unbounded inverse, 
whereas  a lower  semi-frame  has an unbounded frame operator, with bounded inverse.
  For upper semi-frames, in the discrete and the continuous case, we build two natural
    Hilbert scales which may yield a novel characterization of certain function spaces of interest in signal processing. 
We present some examples and, in addition, some results concerning the  duality between lower and upper semi-frames, as well as some generalizations,
including fusion semi-frames and Banach semi-frames.
\end{abstract}

\noi 
\\[2mm]
\noi {\bf AMS Subject Classification}  42C15, 42C40; 46C50; 47A70; 65T60.

\section{\sc Introduction}
\label{sec:intro}

Given a separable Hilbert space $\Hil$, one often  needs to expand an arbitrary element $f\in \Hil $ in a sequence of simple, 
basic elements (atoms) $\Psi = (\psi_k), \, k\in \Gamma$, {with $ \Gamma $ a countable index set:
\be\label{eq:expansion}
f =  \sum _{k\in \Gamma} c_{k} \psi_k , 
\en}
where the sum converges in an adequate fashion (e.g. strongly and unconditionally) and the coefficients $c_{k}$ are (preferably) unique and easy to compute. 
There are several possibilities for obtaining  that result. In order of increasing generality, we can require that $\Psi$ be:
\bei
\vspace*{-2mm}\item [(i)] an orthonormal basis: the coefficients are unique, namely, $c_{k}=\ip{\psi_k}{f}$,  the  convergence is unconditional;  
\vspace*{-2mm}\item [(ii)]  a Riesz basis, \emph{i.e.},  $\psi_k = Ve_{k}$, where $V$ is  bounded bijective operator; the coefficients are unique, 
namely,  $c_{k}=\ip{\phi_k}{f}$, 
where   $(\phi_k)$ is a unique Riesz basis dual to $(Ve_{k})$; the  convergence is unconditional;  
\vspace*{-2mm}\item [(iii)] a frame, that is,  there exist constants  ${\sf m}>0$ and ${\sf M}<\infty$ such that
\be
{\rm\sf m} \norm{}{f}^2 \le  \sum _{k\in \Gamma} \left| \ip {\psi_k}{f} \right|^2  \le {\rm\sf M}  \norm{}{f}^2 ,  \forall \, f \in \Hil.
\label{eq:discr-bddframe}
\vspace*{-4mm}\end{equation}
In this case, uniqueness is    lost.
\eni

However, even a frame may be too restrictive, in the sense that it may impossible to satisfy the two frame bounds simultaneously.
Accordingly, we define $\Psi$ to be an \emph{upper (resp. lower) semi-frame} if it is a total set and satisfies the upper (resp. lower) frame inequality.
Then the question  is to find whether the signal can still be reconstructed from its expansion coefficients.

The notion of frame was introduced in 1952 by Duffin and Schaefer \cite{duffschaef1} in the context of nonharmonic analysis. It was revived by 
Daubechies, Grossmann and Meyer \cite{daub-painless} in the early stages of wavelet theory and then became a very popular topic,
 in particular in Gabor and wavelet analysis  \cite{Casaz1,ole1,daubech1,Groech1}. The reason is that a good frame in a Hilbert space is almost as good as an orthonormal  basis for expanding arbitrary elements (albeit non-uniquely) and is sometimes available while the latter is not (e.g. for continuous wavelets \cite{daubech1}).  
 
 All the above concerns sequences, as required in numerical analysis. However, in the meantime, more general objects, called continuous frames, 
emerged in the context of the theory of (generalized) coherent states  and were thoroughly studied by Ali, Gazeau and the first author  
  \cite{jpa-sqintegI,jpa-contframes,jpa-CSbook} (they were introduced independently by Kaiser \cite{kaiser}).
They were studied further by a number of authors, sometimes under a different name, for instance Fornasier and Rauhut
\cite{forn-rauhut},  Rahimi \emph{et al.} \cite{rahimi} or Gabardo and Han \cite{gabardo-han} (see \cite{jpa-bal} for additional  references).

 Let $\Hil$ be a  Hilbert space and   $X$ a  locally compact space with measure $\nu$. Then a \emph{generalized frame} for $\Hil$  is a family of vectors  
   $\Psi:=\{\psi_{x},\, x\in X\}, \,\psi_{x}\in \Hil $,    indexed by points of $X$, 
such that the map  $x \mapsto \ip{f}{\psi_{x}}$  is measurable,  $\forall\,f \in \Hil$, and
\be\label{eq:genframe}
\int_{X}\ip{f}{\psi_{x}}\ip {\psi_{x}}{f'}\, \ud \nu(x) = \ip{f}{Sf'}, \; \forall\, f, f' \in \Hil,
\en
where $S$ is a bounded, positive, self-adjoint, invertible operator on $\Hil $, called the \emph{frame operator}. 

The operator $S$ is invertible, but its inverse $S^{-1}$, while  self-adjoint and positive, need not be bounded.
We say that   $\Psi$ is a \emph{frame}
\footnote{Several  authors, e.g. \cite{gabardo-han},
  call \emph{frame} the map $\breve\psi: X\to \Hil$ given as $ \breve\psi (x) =  \psi_{x}$}
 if $ S^{-1}$ is bounded or, equivalently, if there exist constants  ${\sf m} > 0$  and ${\sf M}<\infty$
 (the  frame bounds) such that 
\be\label{eq:frame}
{\sf m}  \norm{}{f}^2 \leq \ip{f}{Sf} = \int_{X}  |\ip{\psi_{x}}{f}| ^2 \, \ud \nu(x) \leq {\sf M}  \norm{}{f}^2 ,  \forall \, f \in \Hil.
\en
These definitions are completely general. In particular, if $X$ is a discrete set with $\nu$ being a counting measure, 
we recover the standard definition \eqref{eq:discr-bddframe} of a (discrete) frame \cite{Casaz1,ole1,duffschaef1}.

However, in the general case also, there are situations where the notion of frame is too restrictive, in the sense that one cannot satisfy \emph{both} 
frame bounds simultaneously.
 Thus there is room for two natural generalizations, namely,  we  say that
a family  $\Psi$ is   an \emph{upper (resp. lower) semi-frame}, if
 \bei
\vspace*{-1mm}\item[(i)]  {it is total in $\Hil$;}
\vspace*{-2mm}\item[(ii)] {it satisfies the upper (resp. lower) frame inequality in \eqref{eq:frame}.}
\eni
\vspace*{-1mm} Note that the lower frame inequality automatically implies  that the family is total, \emph{i.e.},  (ii) $\Rightarrow$ (i) for a  lower semi-frame.
Also, in the upper case, $S$   is bounded and $S^{-1}$ is unbounded, whereas, in the lower case,  $S$   is unbounded and $S^{-1}$ is bounded.

In the sequel, we shall study these notions, first in the discrete case, then we shall examine how they can be extended to the general case. 
A number of further generalizations will be addressed at the end. A comprehensive analysis of  semi-frame theory may be found in our previous work \cite{jpa-bal}, 
to which we refer for details.  In particular, we have omitted here all the proofs, keeping only what is needed for the paper to be self-contained. 
A striking new result, as compared to  \cite{jpa-bal} is that every upper semi-frame $\Psi$ generates a natural Hilbert scale 
$\{\H_n, n \in \ZN\}$, with $\H_0 = \Hil$, corresponding to the successive powers  $S^{-1/2}$, and another one, $\{\h_n, n \in \ZN\}$,
generated from the first one by the analysis operator. In addition, $\Psi$ yields upper semi-frames in all spaces $\H_n$. An interesting question is whether one can 
identify explicitly the end spaces of the scales and thus obtain a novel characterization of certain sequence or function spaces. As we will see, the answer is positive for
Schwartz' space of fast decreasing sequences.
 
Before proceeding, we list our definitions and  conventions. The framework is
 a (separable) Hilbert space $\Hil$, with the inner product $\ip{\cdot}{\cdot}$ linear in the first factor.
Given an operator $A$ on $\Hil $, we denote its domain by $\dom(A)$, its range by $\ran(A)$  or, shorter, $R_A$, and its kernel by ${\sf Ker}(A)$.
 An operator $A$ in $\Hil$ is called \emph{positive}, resp.  \emph{nonnegative}, if $\ip{h}{Ah}> 0$,  resp. $\ip{h}{Ah} \geq 0, \, \forall\, h\neq 0, h\in \dom(A)$.
We call an operator $A$ \emph{invertible}, if it is invertible as a function from $\dom(A)$ to $\ran(A)$, \emph{i.e.},  if it is injective. 
$GL(\Hil)$ denotes the set of all bounded operators on $\Hil$ with bounded inverse.

The paper is organized as follows. After a brief summary of the main results about discrete frames  (Section 2), we analyze in detail in Section 3 the properties of discrete semi-frames,
 upper and lower, in particular the Hilbert scales generated by a given upper semi-frame.  In Section 4 we consider two generalizations of frames, namely, fusion frames and Banach frames, 
and analyze how they can be extended to semi-frames (this section has, by necessity, a review character). In Section 5, we briefly summarize the results concerning general 
(`continuous') semi-frames,
 referring to \cite{jpa-bal} for a thorough analysis. In the final Section 6, we reconsider various notions of frame equivalence and their extension to semi-frames.

\section{\sc Discrete frames}
\label{sec: discreteframes}

In the discrete case, we are interested in expansions with norm convergence, thus all the expansions in this section and the next one should be understood as norm convergent.

Let $\Psi =(\psi_k, \,k \in \Gamma)$ be  a frame  for $\Hil$, where $\Gamma$ is some index set (usually $\NN$). 
To this frame $\Psi$, we associate the following three bounded operators:
\bei
\item The analysis operator $C: \Hil \rightarrow  \ell  ^2 $ given by  $Cf = \{\ip {\psi_k }{f }, k \in \Gamma\}$;  

\item The synthesis operator  $D : \ell  ^2   \rightarrow \Hil$  given by
 $Dc = \sum_{k\in \Gamma} c_k \psi_k,$ where $c= (c_k)$;
 
\item  The frame operator $S:\Hil\to\Hil$ given by   
 $Sf = \sum_{k\in \Gamma} \ip{\psi_k}{f}\, \psi_k$,  so that
$$
 \ip{f}{Sf}\ = \sum _{k\in \Gamma} \left| \ip {\psi_k}{f} \right|^2.
$$
\eni
Moreover,   we have $D=C^\ast, \; C=D^ \ast $, and  $S= C^ \ast C$, so that $S$ is self-adjoint and positive, with bounded, self-adjoint inverse $S^{-1}$.

 For the sake of completeness and comparison with the semi-frame case, it is worthwhile to quickly summarize the salient features of  frames.
We do it in the form of a theorem. Of course, all the statements below are
 well-known \cite{Casaz1,ole1,Groech1}, but the approach is non-standard and follows the continuous formalism developed in
 \cite{jpa-sqintegI,jpa-contframes,jpa-CSbook}.
For a proof of this theorem, as well as all the other results quoted in this section, we refer to \cite{jpa-bal}.
\begin{theorem} 
\label{theor-41} 

Let $\Psi=(\psi_k)$ be a frame in $\Hil$,
with analysis operator $C : \Hil \rightarrow \ell^2 $, synthesis operator $D: \ell  ^2   \rightarrow \Hil$ and frame operator is $S: \Hil \rightarrow \Hil $.
Then
\bei
\item[(1)] $\Psi  $ is total in $\Hil$. The operator $S$ has a bounded inverse $S^{-1}:\Hil\to\Hil$ and one has the reconstruction formulas 
\begin{align} \label{srepr}
f &=  S^{-1} S f = \sum _{k\in \Gamma} \ip{\psi_k}{ f}  S^{-1} \psi_k , \; \mbox{ for every }\; f\in \Hil, 
\\[1mm] \label{srepr2}
f &=   S S^{-1} f = \sum _{k\in \Gamma} \ip{S^{-1} \psi_k}{ f}   \psi_k , \; \mbox{ for every }\; f\in \Hil. 
\end{align}

\vspace*{-2mm} \item [(2)] $R_{C}$ is a closed subspace of $\ell^2$. The projection $P_\Psi$ from $\ell^2$ onto  $R_{C}$ is given by $P_\Psi = C S^{-1} D$
 $=C C^+$, where, as before, $C^+$ is the pseudo-inverse of $C$.   

\item [(3)] Define
\be
\ip {c}{d }_{\Psi} = \ip {c}{C S^{-1} C^{-1} d} _{\ell  ^2 }, \,  \, c, d\in R_C. 
\label{eq:discr-scalarprod}
\end{equation}
The relation (\ref{eq:discr-scalarprod}) defines an inner product on $R_C$ and $R_C$ is complete in this inner product. 
Thus, $(R_{C}, \<\cdot,\cdot\>_\Psi)$ is a Hilbert space, which will be denoted by $\h_\Psi $.

\item [(4)] $\h_\Psi $ is a reproducing kernel Hilbert space with kernel given by the matrix  ${\mathcal G}_{k,l} = \ip  { \psi_k}   {S^{-1}\psi_l }$. 

\item [(5)]  The analysis operator $C$ is a unitary operator from $\Hil$ onto  $\h_\Psi $.
Thus, it can be inverted on its range by its adjoint, which leads to the reconstruction formula (\ref{srepr}).  
\eni
\end{theorem} 
{The sequence $(\widetilde \psi_k)$ where $\widetilde \psi_k:= S^{-1} \psi_k$, is again a frame, called the \emph{canonical dual} of $( \psi_k)$, 
because the   relation \eqref{srepr2} means precisely that it is  dual to $( \psi_k)$.  In general, however, a frame may have many different duals, and it has a unique dual
 if and only if  it is a Riesz basis \cite[Cor. 6.65]{casaz-han}. We will extend this notion of duality to semi-frames in Section \ref{subsec:simplexamples}.

\section{\sc Discrete semi-frames}
\label{sec:discsemiframes}

\subsection{\sc Discrete upper semi-frames}
\label{subsec:unbdd-discrframe}

Let now $\Psi$ be an \emph{upper semi-frame}, that is, a sequence $(\psi_k)$ satisfying the relation
\be
0 < \sum_{k\in \Gamma} \left| \ip {\psi_k}{f} \right|^2  \le {\rm\sf M}  \norm{}{f}^2 ,  \forall \, f \in \Hil, \, f\neq 0.
\label{eq:discr-unbddframe}
\end{equation}
{If only the upper bound inequality holds, $\Psi$ is called a \emph{Bessel sequence}.}
Thus an upper   semi-frame is nothing but a total Bessel sequence.  
By analogy with the frame case, we may also introduce a \emph{weighted} upper semi-frame, with weights $v_k\neq 0$,
defined by the obvious relation
$$
0 < \sum_{k\in \Gamma}v_k^2 \left| \ip {\psi_k}{f} \right|^2  \le {\rm\sf M}  \norm{}{f}^2 ,  \forall \, f \in \Hil, \, f\neq 0.
$$
Since all the statements apply to this case also, we will not mention it any more in the sequel.

We begin by analyzing the three operators $C, D$ and $S$, defined exactly as above for a frame  (see \cite{jpa-bal} for a proof). 

\begin{lem.}\label{lem31}
 Let   $\Psi$ be an upper semi-frame. Then one has:
\bei
\item[(1)] The analysis operator $C$ is injective and bounded.   The synthesis operator $D=C\ha$ is   bounded   with dense range.
The frame operator $S$ is   bounded, self-adjoint, positive   with dense range. 
Its inverse $S^{-1}$ is  densely defined and self-adjoint.

\item[(2)] $R_C^{\Psi} \subseteq R_C \subseteq \ov{R_C}$, with dense inclusions, where $R_C^{\Psi}:= C(R_{S})$ and
$\ov{R_C}$ denotes the closure of $R_C$ in $ \ell  ^2  $.
\eni
\end{lem.}
\medskip

At this point, we   make a distinction, that will simplify some statements below.  Namely, we say that the upper semi-frame $\Psi = (\psi_k)$ is \emph{regular}
 if every $\psi_{k}$ belongs to $\dom(S^{-1})=R_S$. First note that, if $\Psi$ is an upper semi-frame for $\Hil$, then 
 \begin{equation} \label{ssminus1}
 f=SS^{-1}f= \sum _{k\in \Gamma} \ip{\psi_k}{S^{-1}f}\psi_k, \ \forall f\in R_S.
 \end{equation}
If we want to write the expansion above using a dual sequence (similar to the frame expansion \eqref{srepr2}), then the upper semi-frame should be regular. 
Indeed, since $S$ is bounded and $S^{-1}$ is self-adjoint, we have
\begin{prop.} \label{ssminus} 
Let $\Psi$ be a regular upper semi-frame for $\Hil$. Then
\begin{equation}\label{eq:trivreconstunb1}
f=SS^{-1}f=\sum _{k\in \Gamma} \ip{S^{-1}\psi_k}{f}\psi_k, \ \forall f\in R_S.
\end{equation}
\end{prop.} 
\medskip

However, it is not  possible to extend this reconstruction formula to all $f\in\Hil$ by a limiting procedure.  
If the reconstruction formula can be extended in the strong sense to the whole Hilbert space, then the original sequence was already a frame.  
A detailed argument to that effect  is given in \cite[Remark 3.4.1]{jpa-bal}. For a reconstruction formula   in the weak  sense, see Section \ref{subsec:discG-triplet}.

\medskip

The whole motivation of the present construction is to translate abstract statements in $\Hil$ into concrete ones about sequences,
 taking place in $\ell^2$. The correspondence is implemented by the operators $C$ and $C^{-1}$. Hence we first transport
the operators $S$ and $S^{-1}$, according to the following proposition.

\begin{prop.} \label{prop44}
\bei\item[(1)] Define the operator $G_{S} : R_C \to C(R_S) $   by  $G_{S} = C   S Ê  C^{-1}$  
and the operator $G_{S} ^{-1} :  C(R_S) \to R_C$   by  $G_{S} ^{-1} = C   S^{-1} Ê  C^{-1}{|_{C(R_S)}}.$ Then, in the Hilbert space $ \ov{R_C}$,
$G_{S} $ is a bounded, positive and symmetric operator, while $G_{S} ^{-1}$ is positive and essentially self-adjoint.
These two operators are bijective and inverse to each other.

\item[(2)]  Let $G = \ov{G_{S}}$ and let $G^{-1}$ be the self-adjoint extension of  $G_{S} ^{-1}$.
Then $G: \ov{R_C} \to R_{G} \subseteq \ov{R_C}$ is 
bounded, self-adjoint and positive  with  $\dom(G) = \ov{R_C}$, thus $G = C D{|_{ \ov{R_C}}}$.
Furthermore $G^{-1}: \dom(G^{-1}) \subset \ov{R_C} \to \ov{R_C}$ is self-adjoint and positive, with domain 
$ \dom(G^{-1}) = R_{G}=   C (R_{D})  $, a dense subspace of $ \ov{R_C}$. The two operators are inverse of each other, in the sense that
$$
 G^{-1}G = I{|_{\ov{R_C}}}, \quad G G^{-1}=  I{|_{C(R_D)}}.
$$
\eni
\end{prop.}
A proof of Proposition \ref{prop44}  may be found in \cite[Prop.3.4]{jpa-bal},  partly  following   the similar one in the original paper \cite{jpa-sqintegI}, 
which concerned the case of generalized frames. See also Section \ref{sec:contframes}.
\medskip

Putting everything together, we have  the following commutative diagram: 
\be\label{diagram3}
\hspace*{-1cm}\begin{minipage}{15cm}
\begin{center}
\setlength{\unitlength}{2pt}
\begin{picture}(100,60)
\put(10,50){$\Hil$}
\put(89,50){$R_C \subseteq  \ov{R_C}\subseteq \ell^2 $}
\put(20,52){\vector(4,0){65}}
\put(54,44){$C^{-1}$}
\put(85,50){\vector(-4,0){65}}
\put(55,54){$C$}
\put(10,47){\vector(0,-4){30}}
Ê\put(12,17){\vector(0,4){30}}
\put(91,47){\vector(0,-4){30}}
Ê\put(93,17){\vector(0,4){30}}
\put(-35,10){$ \Hil \supseteq \dom(S^{-1})=R_S$}
\put(85,47){\vector(-2,-1){65}}
\put(14,30){$S^{-1}$}
\put(5,30){$S$}
\put(82,30){$G_{S}$}
\put(95,30){$G_{S} ^{-1}$}
\put(86,10){$ C(R_S)\subseteq \ell^2$}
\put(19,12){\vector(4,0){65}}
\put(50,4){$C^{-1}$}
\put(48,33){$D$}
Ê\put(84,10){\vector(-4,0){65}}
\put(50,15){$C$}
\end{picture}
\end{center}
\end{minipage}
\end{equation}

As before, define on $C(R_S)$ the new inner product $\ip {c}{d}_ \Psi = \ip{c}{G^{-1} d}_{\ell^2 }$, which makes sense
since $G^{-1}$ is self-adjoint and positive. Therefore $\ip {c}{d}_ \Psi = \ip{G^{-1/2}c}{G^{-1/2} d}_{\ell^2 }$.
Denote by ${\h}_{\Psi}$ 
the closure of $C(R_S)$ in the corresponding  norm $\norm{\Psi}{\cdot}$, which is a Hilbert space.

Then the fundamental result reads as follows  (a proof is given in \cite[Theorem 3.6]{jpa-bal}).
\begin{theorem}
\label{theor-43} Let ${\h}_{\Psi}$ be the closure of $C(R_S)$ in the new norm $\norm{\Psi}{\cdot}$. 
Then: 
\bei
\item[(1)] ${\h}_{\Psi}$ coincides with $R_C$ (as sets) and $C$ is a unitary map  (isomorphism) between $\Hil$ and ${\h}_{\Psi}$.

\item[(2)] 
The norm $\norm{\Psi}{.}$ is equivalent to the graph norm of $G^{-1/2}$ and, therefore, $\dom( G^{-1/2} ) = {\h}_{\Psi}$. 

\item[(3)]
  $C^{*(\Psi)} = \left( S^{-1} D \right)\!{|_{{\h}_{\Psi}}}$, where $C^{*(\Psi)}:{\h}_{\Psi} \to \Hil$ is the adjoint of
$C: \Hil  \to{\h}_{\Psi} $. Moreover,  for every $f\in\Hil$, one has
$$
 f = C^{*(\Psi)} C f = \left( S^{-1} D \right) C f. 
$$
\item[(4)] 
${G^{1/2}}: \ov{R_C} \rightarrow {\h}_{\Psi}$ is a unitary operator  and so is its inverse  ${G^{-1/2}} : Ê{\h}_{\Psi} \rightarrow \ov{R_C}$. 
\eni
\end{theorem} 

\vspace*{5mm}
Thus we have the following diagram that extends \eqref{diagram3}:
\be\label{diagram4}
\begin{array}{cccc}
\Hil       &\stackrel{C}{\longrightarrow }&{\h}_{\Psi}= R_{C} \subset & \!\!\!\!\ov{R_{C}} \subset \ell^{2} \\[1mm]
\cup&                            &\cup& \\[1mm]
\dom(S^{-1}) = R_{S}   &\stackrel{C}{\longrightarrow }  &\qquad C(R_S)\quad \subset & \hspace{-3mm}  \ell^{2}
\end{array}
\end{equation}

As expected, the regularity of an upper semi-frame allows us to derive results analogous to those obtained for a frame in Theorem \ref{theor-41}, namely,
\begin{prop.}\label{sec:rkhsunbound1} 
Let $(\psi_k)$ be a regular upper semi-frame.
Then ${\h}_{\Psi}$ is a reproducing kernel Hilbert space, with kernel given by  the operator $S^{-1}D $, which is a matrix operator,
 namely,  the matrix $\mathcal G$, where  $\mathcal G_{k,l} = \ip{Ê\psi_k}   {S^{-1} \psi_l}   = \ip {\psi_k}  { C^{-1} G^{-1} C Ê\psi_l}  $. 

\end{prop.}

{For $f \in R_S$ we have $f = S S^{-1} f$. So, for a regular upper semi-frame, we can give the reconstruction formula  \eqref{ssminus1} for all $f \in R_S$, 
which reads as:
$$ f = \sum _{k\in \Gamma} \ip {\widetilde \psi_k}{f} \psi_k,
 $$ 
where, as usual, $\widetilde \psi_k:= S^{-1} \psi_k$ denotes the canonical dual.
Other  reconstruction formulas  may be given  for every $f\in R_D$ and even    
 for all $f \in \Hil$, even in the case  of a nonregular upper semi-frame,
 if we allow the analysis coefficents to be altered. However, the resulting formulas are not very useful
 since they use an operator-based approach and don't use sequences for the inversion. Hence we skip this, referring instead to
 \cite[Theorem 3.6(4) and Theorem 3.8]{jpa-bal}.}

For a treatment of the existence of dual sequences  and related questions, we refer to \cite{jpa-Besseq}.

\subsection{\sc Formulation in terms of a Gel'fand triplet}
\label{subsec:discG-triplet}

If the upper semi-frame $\Psi$ is not regular, we  have to turn to distributions, using for instance the language of   Rigged Hilbert spaces or Gel'fand triplets  \cite{gelf}, 
as we show now.  Actually, we get here a simpler version, namely a triplet of Hilbert spaces, the simplest form of (nontrivial) partial inner product space  \cite{jpa-pipspaces}.

The construction  goes back to  \cite[Section 4]{jpa-contframes} and \cite[Section 7.3]{jpa-CSbook}, in the general case (see also \cite{jpa-bal}).
 When particularized to the discrete environment, the argument goes as follows. 
If  $\Psi$ is regular, the reproducing matrix $\mathcal G$, introduced in Proposition \ref{sec:rkhsunbound1},  defines a bounded sesquilinear form over ${\h}_{\Psi}$, namely,
\be\label{eq:discrsesqform1}
K^\Psi(d,d') := \sum _{k,l\in \Gamma} \ov{d_{k}}\,\mathcal G_{k,l}\,d'_{l} = \ip{C^{-1} d}{SC^{-1}d'}_\Hil, \; \forall\, d,d'\in {\h}_{\Psi}.
\end{equation}
However, the resulting relation still makes sense  even  if  $\Psi$ is not regular, that is, 
\be\label{eq:discrsesqform2}
K^\Psi(d,d') = \ip{C^{-1} d}{SC^{-1} d'}_\Hil, \; \forall\, d,d'\in {\h}_{\Psi}.
\end{equation}
Denote by ${\h}_{\Psi}^{\times}$ the Hilbert space obtained by completing  ${\h}_{\Psi}$ in the norm given by this sesquilinear form.
Now, \eqref{eq:discrsesqform1} and \eqref{eq:discrsesqform2}  imply that 
$$ 
 K^\Psi(d,d') = \ip{C^{-1} d}{SC^{-1} d'}_\Hil = \ip{d}{C SC^{-1} d'}_\Psi = \ip{d}{d'}_{\ell^2} .
$$
Therefore, one obtains, with continuous and dense range embeddings,
\be
 {\h}_{\Psi} \;\subset\; {\h}_{0}= \ov{{\h}_{\Psi}}\;\subset \; {\h}_{\Psi}^{\times},
\label{eq:disctriplet}
\end{equation}
where 
\begin{itemize}
\vspace{-1mm}\item[{\bf .}]  $ {\h}_{\Psi} = R_{C}$, which is a Hilbert space for the norm 
$\norm{\Psi}{\cdot}=\ip {\cdot}{ G^{-1}  \cdot }^{1/2}_{\ell^2}$; 
\vspace{-1mm}\item[{\bf .}]   ${\h}_{0} =\ov{{\h}_{\Psi}} = \ov{R_{C}}$ is the closure of ${\h}_{\Psi}$ in $\ell^2$;
\vspace{-1mm}\item[{\bf .}]  
${\h}_{\Psi}^{\times}$ is 
the completion of  ${\h}_{0}$  (or ${\h}_{\Psi}$) in the norm $\norm{\Psi^{\times}}{\cdot}:=\ip {\cdot}{G \cdot }^{1/2}_{\ell^2}$.
\end{itemize}

 The notation in \eqref{eq:disctriplet} is coherent, since  the space $\h^{\times}_{\Psi}$ just constructed is the conjugate dual of ${\h}_{\Psi}$, 
 \emph{i.e.},  the space of conjugate linear functionals on   ${\h}_{\Psi}$ 
(we use the conjugate dual instead of the dual, in order that all embeddings in \eqref{eq:disctriplet} be linear). 
Indeed, since $K^\Psi$ is a bounded sesquilinear form over  ${{\h}_{\Psi}}$, the relation $ X_d= \ov {K^\Psi(d,\cdot)}$
 defines, for each $d\in{\h}_{\Psi}$,  an element $X_d$ of the conjugate dual of  ${\h}_{\Psi}$ (note that  $X_d$ depends \emph{linearly} on $d$). 
If, on these elements, we define the inner product
 $$
 \ip{X_d}{X_{d'}}_{{\Psi}^{\times}} = \ip{C^{-1} d}{SC^{-1} d'}_\Hil = K^\Psi(d,d')
 $$
 and take the completion, we obtain precisely the Hilbert space ${\h}_{\Psi}^{\times}$. Thus \eqref{eq:disctriplet} is a Rigged Hilbert space or a Gel'fand triplet, 
 more precisely a Banach (or Hilbert) Gel'fand triple in the terminology of Feichtinger \cite{fei-zimm}.
 
The sesquilinear form $K^\Psi$ gives some way of inverting the analysis operator, as follows.
Given $f\in \Hil, \, d=Cf \in {\h}_{\Psi}$, consider the relation  
$$
 \ip{f}{f'}_{\Hil}=\ip{d}{CS^{-1}f'}_{\ell^2}, \; f\in \Hil, \, f'\in R_{S}.
$$
Define $d':=CS^{-1}f' \in {\h}_{\Psi}$. 
Even if $\Psi$ is not regular, we can associate to  $d'$   an element  $X_{d'}\in {\h}_{\Psi}^{\times}$, namely,
$$
X_{d'}(d) = K^\Psi(d,d') = \ip{d}{d'}_{\ell^2} = \ip{f}{f'}_{\Hil}.
$$
However, this procedure does not give explicit reconstruction formulas.

In the triplet \eqref{eq:disctriplet}, ${\h}_{0}$ is a sequence space contained in $\ell^2$ (possibly $\ell^2$ itself), ${\h}_{\Psi}$ is a smaller sequence space, for instance a space
of \emph{decreasing} sequences. Then  ${\h}_{\Psi}^{\times}$ is the K\"{o}the dual of ${\h}_{\Psi}$, normally not contained in $\ell^2$ \cite{lux}. 
In the example below (Section \ref{subsec:simplexamples}), ${\h}_{\Psi}^{\times}$ consists of slowly \emph{increasing} sequences.

Now, if $\Psi$ is regular, all three spaces $ {\h}_{\Psi},{\h}_{0}, {\h}_{\Psi}^{\times}$ are  {reproducing kernel Hilbert spaces}, with the same 
(matrix) kernel  ${\mathcal G}(k,l) = \ip{\psi_{k} }{S^{-1}\psi_{l}}.$ 

Finally, if $S^{-1}$ is bounded, that is, in the case of a frame, the three Hilbert spaces of \eqref{eq:disctriplet} coincide as sets, with equivalent norms,
 since then both $S$ and $S^{-1}$  belong to $GL(\Hil)$. 

A possibility to have a more general reconstruction formula is the following. In the triplet \eqref{eq:disctriplet}, the `small' space ${\h}_{\Psi}$ is the  domain of 
$ G^{-1/2}$, with norm $\norm{\Psi}{\cdot}=\ip {\cdot}{ G^{-1}  \cdot }^{1/2}_{\ell^2}$.  On the side of $\Hil$, \emph{i.e.},  on the left-hand side of the diagram \eqref{diagram4},
this corresponds to the   domain of $S^{-1/2} $, with norm ${\norm{\Psi}{\cdot}}\!\!\!^{\widetilde \;\;} =\ip {\cdot}{S^{-1}  \cdot }^{1/2}_{\Hil}$. 
We can  consider instead the smaller space 
  $R_S = \dom(S^{-1}) $, with norm $\norm{\s}{\cdot}=\ip {S^{-1} \cdot}{S^{-1}  \cdot }^{1/2}_{\Hil}$. The resulting space, denoted $\s$, is complete, hence a Hilbert space.
Thus, adding the conjugate dual $ \s^{\times}$ of $ \s$,  we get a new triplet
\be\label{eq:newtriplet}
 \s \subset \Hil \subset  \s^{\times}.
\end{equation}
In the new triplet \eqref{eq:newtriplet}, the operator $S^{-1}$ is isometric from
$ \s$ onto $\Hil $ and, by duality, from $\Hil $ onto $ \s^{\times}$ .
The benefit of that construction is that the relation \eqref{ssminus} is now valid for any  $f\in \Hil$,
 even if $\Psi$ is non-regular, but of course, we still need 
that $f'\in \s = R_{S}$. In other words, we obtain  a reconstruction formula in the sense of distributions, namely,
\be\label{ssminus2}
\ip{f}{f'}= \sum _{k\in \Gamma} \ip{f}{\psi_k} \ip{S^{-1}\psi_k}{f'}\psi_k, \ \forall   f\in\Hil, \,\forall f' \in \s=R_S.
\end{equation}
Mapping everything into $\ell^2$ by $C$, we obtain the following scale of  Hilbert spaces:
\be\label{eq:multiplet}
\h_{\s} \subset \h_{\Psi} \subset \h_{0} = \ov{R_C} \subset {\h}_{\Psi}^{\times} \subset {\h}_{\s}^{\times}.
\end{equation}
In this relation,   $\h_{\s}:= C  \s=  C(R_{S})$, with norm\footnote{The expression for the norm of $\h_{\s}$ given  in \cite{jpa-bal},  after Eq.(2.16),
 is not correct.} 
 $\norm{\s}{\cdot}=\ip {\cdot}{ G^{-3}  \cdot }^{1/2}_{\ell^2}$.

\subsubsection{\sc A scale of Hilbert spaces}

Combining these results with the diagram (3.15) of \cite{jpa-bal} and extending the latter, we obtain the following scheme (note that the previous diagram 
has been inverted, both horizontally and vertically):
\bigskip

\be\label{gramframconn1}
 \vspace*{1cm}
\hspace*{-2cm}\begin{minipage}{15cm}
\begin{center}
\setlength{\unitlength}{1.9pt}
\begin{picture}(100,60)
\put(-25,50){$\cdots$} 

\put(-15,50){$\stackrel{S^{-1/2}}{\longrightarrow }$}
\put(7,63){$\H_{2}  $}   
\put(10,57){$\shortparallel  $}   
\put(0,50){$R_S= \s$} 

\put(23,50){$\stackrel{S^{-1/2}}{\longrightarrow }$}
\put(40,63){$\H_{1}  $}   
\put(42,57){$\shortparallel  $}   
\put(40,50){$R_D$} 

\put(55,50){$\stackrel{S^{-1/2}}{\longrightarrow }$}
\put(72,63){$\color{red}\H_{0}  $}   
\put(74,57){$\color{red}\shortparallel  $}   
\put(72,50) {$\color{red}\Hil $} 

\put(86,50){$\stackrel{S^{-1/2}}{\longrightarrow }$}
\put(105,63){$\H_{-1}  $}   
\put(109,57){$\shortparallel  $}   
\put(105,50){$R_{S^{-1/2}}$}  

\put(125,50){$\stackrel{S^{-1/2}}{\longrightarrow }$}
\put(150,63){$\H_{-2}  $}   
\put(153,57){$\shortparallel  $}   
\put(140,50){$R_{S^{-1}}= \s^{\times}$} 

\put(173,50){$\stackrel{S^{-1/2}}{\longrightarrow }$}
\put(188,50){$\cdots$} 

\put(-50,10){$\stackrel{G^{-1/2}}{\longrightarrow }$}
\put(-60,10){$\cdots$}

\put(-35,10){$C(R_S)  $}
\put(-35,3){$\quad \shortparallel  $}   
\put(-39,-5){$  {\h}_{\s} \equiv {\h}_{3}  $}   

\put(-15,10){$\stackrel{G^{-1/2}}{\longrightarrow }$}
\put(0,10){$ C(R_D)$}  
\put(5,3){$ \shortparallel  $}   
\put(4,-4){$  {\h}_{2}  $}   

\put(22,10){$\stackrel{G^{-1/2}}{\longrightarrow }$}
\put(40,10){$ R_C $}
\put(42,3){$\shortparallel  $}   
\put(34,-4){$  {\h}_{\Psi}\equiv {\h}_{1}  $}   

\put(52,10){$\stackrel{G^{-1/2}}{\longrightarrow }$}
\put(70,10){$\color{red}\ov{R_C}  $}   
\put(72,3){$\color{red} \shortparallel  $}   
\put(71,-4){$\color{red}  {\h}_{0} $}   

\put(87,10){$\stackrel{G^{-1/2}}{\longrightarrow }$}
\put(105,10){$C(\s^{\times})$}

\put(109,3){$  \shortparallel  $}   
\put(100,-4){$ {\h}_{\Psi}^{\times}\equiv {\h}_{-1} $}   
\put(125,10){$\stackrel{G^{-1/2}}{\longrightarrow }$}
\put(140,10){$\cdots$}
\put(41,18){\vector(-1,1){28}} 
\put(74,18){\vector(-1,1){28}} 
\put(105,18){\vector(-1,1){28}} 
\put(10,46){\vector(-1,-1){28}} 
\put(41,46){\vector(-1,-1){28}} 
\put(73,46){\vector(-1,-1){28}} 
\put(105,46){\vector(-1,-1){28}} 
\put(140,46){\vector(-1,-1){28}} 
\put(15,33){$\scriptstyle D$}
\put(49,33){$\scriptstyle D$}
\put(82,33){$\scriptstyle D$}
\put(1,33){$\scriptstyle C$}
\put(31,33){$\scriptstyle C$}
\put(66,33){$\scriptstyle C$}
\put(97,33){$\scriptstyle C$}
\put(131,33){$\scriptstyle C$}
\end{picture}
\end{center}
\end{minipage}
\vspace*{-5mm}
\end{equation}

{In the upper row of \eqref{gramframconn1}, the operator $ S^{-1/2}$ is unitary from each space onto the next one.
The same is true for the operator $ G^{-1/2}$ in  the lower row.}

{Actually one can go further. Indeed, in the multiplet \eqref{eq:multiplet}, the space $\h_{\Psi}$ is the domain of $ G^{-1/2}$,  
  and $\h_{\s}$ is the domain of $ G^{-3/2}$, both considered with their graph norm in the topology of $\ov{R_C} \subset \ell^2$. 
Thus the multiplet \eqref{eq:multiplet} is the central part  of the Hilbert scale built on the powers of 
the positive self-adjoint operator $ G^{-1/2}$, namely,  $\h_{n}:= \dom(G^{-n/2}), nÊ\in \ZN : $
$$ 
 {\h}_{\Psi}\equiv \h_{1},  \quad C(R_D) \equiv \h_{2}, \quad {\h}_{\s} \equiv \h_{3}, \quad {\h}_{\Psi^\times} \equiv \h_{-1}, \quad  \ldots.
$$
Similar considerations apply to the triplet \eqref{eq:newtriplet}, 
which is the central part of the scale built on the powers of $S^{-1/2} $, \emph{i.e.},  $\H_{n}:= \dom(S^{-n/2}), nÊ\in \ZN$.
For better visualization, we have highlighted the central space of both scales.} 
 
In both cases, we obtain in this way a simple partial inner product space \cite{jpa-pipspaces}.  A natural question then is to identify the end spaces, 
\be \label{endscale}
\h_{\infty}(G^{-1/2}):=\bigcap_{n\in \ZN} \h_n, \qquad \h_{-\infty}(G^{-1/2}):=\bigcup_{n\in \ZN} \h_{n},
\end{equation} 
and similarly for the scale $\H_{n}:= \dom(S^{-n/2}), nÊ\in \ZN$.  
In the simple examples of Section \ref{subsec:simplexamples}, the question can be answered explicitly. In this way, one has at one's disposal 
the full machinery of partial inner product spaces \cite{jpa-pipspaces}.
For instance, one can ask under which conditions the space $\h_{\infty}$ is nuclear, and similar questions whose answer relies on the structure of the full scale.

In the diagram \eqref{gramframconn1}, on the left side of the central spaces, the operators $C$ and $D$ are defined as usual, and to the right they are defined by duality. 
Clearly, using the above notation we have $D : \h_n \rightarrow  \HH_{n+1}$,  $C : \HH_n \rightarrow \h_{n+1}$ and $S:  \HH_n \to  \HH_{n+2}$. 
Furthermore,  $S^a$ being unitary for every $a\in\RN$,  one has $(S^a)^* = S^{-a}$. 

But we can say   more,  since dual spaces of sequence spaces are sequence spaces again, with the duality relation $\ip{c}{d} = \sum_k \ov{c_k } {d_k}$, 
inherited from the space $\omega$ of \emph{all} sequences (see \cite[Secs. 1.1.3 and 4.3]{jpa-pipspaces}  or \cite[\S 30] {kothe}).
  \begin{theorem}
 \label{theo-scale}
Let $\psi_k \in \HH_{n_0}$ for a $n_0 \ge 0$. Then for all $n \le n_0$ we have: 
\bei
\vspace*{-2mm} \item[(1)]
$C : \HH_{-n} \rightarrow \h_{-(n-1)}$ is given by 
 $$
 C f =  \{\ip {\psi_k}{f}_{ \HH_{n}, \HH_{-n}}, k\in \Gamma\}. 
$$ 
\vspace*{-5mm} \item[(2)] 
 $D : \h_{-n}\rightarrow \HH_{-(n-1)} $ is given by 
$$
 D c = \sum_k c_k \psi_k,
 $$
\vspace*{-3mm}with weak convergence.
\item[(3)]  Let $m \le n_0-2$. Then, for all $f \in \HH_{-m}$, we have the reconstruction formula (in a weak sense):
$$
f = \sum _k\ip {S \psi_k}{f}_{\HH_{m},\HH_{-m}} \psi_k, \quad  \psi_k\in \HH_{m+2}.
$$
\eni
 \end{theorem}
\begin{proof}  Note that, for the chosen $n$, we have $\psi_k \in \HH_n$.
  For $d \in \h_{n-1} $ and $f \in \HH_{-n}$, we have 
$$ 
\ip  {d}{Cf}  _{\h_{n-1}, \h_{-(n-1)}} = 
  \ip{D d}{ f}  _{ \HH_{n}, \HH_{-n}} = \ip { \sum_{k} d_k \psi_k }{f}  _{ \HH_{n}, \HH_{-n}}
=\sum_{k} \ov{d_k}\, {\ip {\psi_k} {f}  }_{ \HH_{n}, \HH_{-n}} .  
 $$
On the other hand, let $c \in \h_{-n}$ and $g \in \HH_{n-1}$. Then, 
$$ 
\ip {D c} {g} _{\HH_{-(n-1)},\HH_{n-1}} = \ip {c}  { C g} _{\h_{-n},\h_{n}} = \sum _k \ov{c_k }\ip {\psi_k}{g} = \ip{\sum  _k c_k \psi_k}{ g}.
$$ 
Note that, for all $n$, one has
$S^{\pm 1} : \HH_{n} \rightarrow \HH_{n\pm 2}$ unitarily, so for $f \in \HH_{-m}$ we have $S^{-1} f \in \HH_{-(m+2)}$.
 Clearly $ f = S S^{-1} f = D C S^{-1} f$. 
 By assumption $m \le n_0-2$ and so we may take $\psi_k \in \HH_{m+2}$. Therefore 
$$
(C S^{-1} f)_{k} = \ip { \psi_k}{S^{-1} f}_{\HH_{m+2},\HH_{-(m+2)}} =\ip {S \psi_k}{f}_{\HH_{m},\HH_{-m}}, 
$$
 since $S: \HH_{m}\to \HH_{m+2}$ is unitary.

 \end{proof}
 \medskip
 
Clearly, a regular frame corresponds to $n_0=2$, that is,  $\psi_k \in \HH_{2}= R_S= \dom(S^{-1}) $. More generally, we say that $\Psi$ is
$n_0$-regular whenever $\psi_k \in \HH_{n_0}$, as in Theorem   \ref{theo-scale} (so that `regular' is `2-regular').
But now, Theorem \ref{theo-scale} suggests that we consider a smoother case, namely, that $\psi_k \in \HH_{\infty}(S^{-1/2}):=\bigcap_{n\in \ZN} \HH_n$. 
Then we   say that $\Psi$ is a \emph{totally regular} upper semi-frame. Clearly the three statements of Theorem \ref{theo-scale} hold now for every $n\in\ZN$, namely, 
\begin{prop.}
Let $\Psi = (\psi_k)$ be a  totally regular upper semi-frame. Then one has, for every $n\in\ZN$,
\bei
\vspace*{-2mm} \item[(1)]
$C : \HH_{-n} \rightarrow \h_{-(n-1)}$ is given by  $C f =  \{\ip {\psi_k}{f}_{ \HH_{n}, \HH_{-n}}, k\in \Gamma\}.$ 
\vspace*{-2mm} \item[(2)]
$D : \h_{-n}\rightarrow \HH_{-(n-1)} $ is given by  $D c = \sum_k c_k \psi_k, $ with weak convergence.
\vspace*{-2mm} \item[(3)]
For all $f \in \HH_{-n}$, we have the reconstruction formula (in a weak sense):
$$
f = \sum _k\ip {S \psi_k}{f}_{\HH_{n},\HH_{-n}} \psi_k, \quad  \psi_k\in \HH_{\infty}.$$
\eni
\end{prop.}
In addition,  the upper semi-frame $\Psi= (\psi_k)$  generates a whole set of  upper semi-frames in  the scale, and even more if it is $n_0$-regular or totally regular Indeed:
\begin{prop.} 
Let  $\Psi = (\psi_k)$ be an  upper semi-frame.  Then the following holds. 
\bei
\vspace*{-1mm} \item[(1)] 
The family $(S^{n/2}\psi_k)$ is an upper semi-frame in every space   $\HH_n, n\in \ZN$.
\vspace*{-1mm} \item[(2)]
Let $\Psi = (\psi_k)$ be    $n_0$-regular. Then, for every $n \leq n_0$, $\Psi$ is an upper semi-frame in the space $\HH_n$.
Similarly, if $\Psi$ is totally regular, the same is true for all $n\in \ZN$.
\eni
\end{prop.}
\begin{proof} (1) Since $S^{-n/2} : \HH_n \to \HH_0 = \H$ is a unitary map, we have, for any  $f\in\HH_n$,
$$
\sum_k |\ip{S^{n/2}\psi_k}{f}_{\HH_n}|^2 \leq \sum_k |\ip{\psi_k}{S^{-n/2}f}_{\HH_0}|^2 \leq B \norm{\HH_0}{S^{-n/2}f}^2 
=B  \norm{\HH_n}{f}^2,
$$
since  $(\psi_k)$  is an upper semi-frame in $\H$. This   indeed shows that $(S^{n/2}\psi_k)$ is an upper semi-frame in $\HH_n$.

(2) Let $\Psi$ be $n_0$-regular. Then, for  $n \leq n_0$ and any $f\in \HH_n$, one has
\begin{align*}
\sum_k |\ip{\psi_k}{f}_{\HH_n}|^2 &\leq \sum_k |\ip{S^{-n/2}\psi_k}{S^{-n/2}f}_{\HH_0}|^2 =  \sum_k |\ip{\psi_k}{S^{-n}f}_{\HH_0}|^2 \\
&\leq B \norm{\HH_0}{S^{-n}f}^2 =B  \norm{\HH_{2n}}{f}^2 \leq B'\norm{\HH_{n}}{f}^2,
\end{align*}
since the embedding $\HH_{2n} \to \HH_{n}$ is continuous.
 \end{proof}
 
In particular, a regular upper semi-frame $(\psi_k)$ in $\H$ is automatically an upper semi-frame in $\HH_2= R_S= \dom(S^{-1}) $.

\berem 
It is noteworthy that this condition of total regularity (for frames, in fact) was already introduced in the context of partial inner product spaces \cite[Sec. 3.4.4]{jpa-pipspaces},
 but in a different perspective. There, indeed, one starts with a given partial inner product space $V_I$ and asks under which conditions a family of vectors $\Psi = (\psi_k)$
 may constitute a frame. The argument runs as follows.  
The purpose of a frame is to expand an arbitrary vector into simple elements, as in  \eqref{eq:expansion}, and in practice this expansion will be truncated after finitely
 many terms for approximation. Now, in a partial inner product space, finite rank projections must have their range in the small space $V^\#$. 
Therefore, one has to require that the frame vectors (or basis vectors, as well) $\psi_k$ must belong to $V^\#$, which is precisely the space $\HH_{\infty}$ in the present context. 
Here, on the contrary, the semi-frame $\Psi$ itself generates the scale, which thus allows much more singular situations.
\enrem 

 \subsection{\sc Lower semi-frames and duality}
\label{subsec:simplexamples} 

To start with,  two sequences $(\psi_k), (\phi_k)$ are said to be \emph{dual} to each other if one has, for every $f\in\Hil$,
\begin{equation} \label{ff}
f=\sum _{k\in \Gamma}\ip {\phi_k}{f}\, \psi_k = \sum _{k\in \Gamma} \ip {\psi_k}{f}\,  \phi_k .
\end{equation}
We are going to explore to what extent this notion applies to upper and lower semi-frames

To be precise, we say that  a    sequence $\Phi=\{\phi_{k}\}$
 is a \emph{lower semi-frame} if  it satisfies the lower frame condition, that is, there exists   a constant ${\sf m}>0$ such that
\be
{\sf m}  \norm{}{f}^2 \le  \sum _k |\ip{\phi_{k}}{f}| ^2 \, , \;\; \forall \, f \in \Hil.
\label{eq:disclowersf} 
\end{equation}
Clearly, \eqref{eq:disclowersf} implies that the family $\Phi $ is total in $\Hil$. Notice there is a slight dissimilarity between the two definitions of semi-frames.
In the upper case \eqref{eq:discr-unbddframe}, the positivity requirement on the left-hand side ensures that the sequence $\Psi$ is total, 
whereas here, it follows automatically from the lower frame bound. 
Before exploring further the duality between the two notions, let us give some simple examples. 

 Let $(e_k), k\in \NN,$ be an orthonormal basis in $\Hil$.
 Let $\psi_k = \frac{1}{k} e_k$. Then $(\psi_k)$ is an upper semi-frame:
$$
0 <  \sum_{k\in \NN} | \ip{\psi_k }{f}|^2 \leq \sum_{k\in \NN} | \ip{e_k }{f}|^2 = \|f\|^2.
$$
Indeed, there is no lower frame bound, because for $f=e_{p}$, one has $\sum_{k\in \NN} | \ip{\psi_k }{f}|^2=\frac{1}{p^2}$.

  Let  $\phi_{k}= k\, e_{k}$. The sequence $(\phi_{k})$ is dual to $(\psi_k)$,  since it obviously satisfies the relations \eqref{ff}.
In addition, we have
$$
\sum_{k\in \NN} | \ip{e_k }{f}|^2 = \|f\|^2 \leq \sum_{k\in \NN} | \ip{\phi_k }{f}|^2,
$$
and this is unbounded since $\sum_{k\in \NN}| \ip{\phi_k }{f}|^2 = p^2$ for $f=e_{p}$. Hence,  $(\phi_{k})$ is a lower semi-frame, dual to  $(\psi_k)$.

  In this case, in the basis $(e_k)$, the frame operator associated to $(\psi_k)$  is $S= \mathrm{diag}({1}/{n^2})$. 
Thus $S^{-1}= \mathrm{diag}(n^2)$, which is clearly unbounded.
It follows that $(\phi_{k})$ is the canonical dual of $(\psi_k)$, since $\phi_{k} = S^{-1} \psi_k $. 
 The sequence used by Gabor  in his original IEE-paper \cite{gabor}, a Gabor system with a Gaussian window, $a=1$ and $b=1$, is exactly 
such an upper semi-frame. 

Similarly, $G^{-1}= \mathrm{diag}(n^2) $, acting in $\ell^2$,  so that the inner products of the three spaces in \eqref{eq:disctriplet} are, respectively:
\begin{itemize}
\vspace{-1mm}\item[{\bf .}] For ${\h}_{\Psi}: \quad\ip{c}{d}_{\Psi} = \sum_{n}{n^2}\, \ov{c_{n}}\,d_{n}$; 
\vspace{-1mm}\item[{\bf .}]   For ${\h}_0: \quad\;\ip{c}{d}_{0} =\sum_{n}  \ov{c_{n}}\,d_{n}$; 
\vspace{-1mm}\item[{\bf .}]   For ${\h}_{\Psi}^{\times}: \quad\ip{c}{d}_{\Psi}^{\times} =\sum_{n} \frac{1}{n^2} \,\ov{c_{n}}\,d_{n}$. 
 \end{itemize}

Both $C\Psi = (C\psi_k)$ and $C\Phi =  (C\phi_k)$ live  in ${\h}_{\Psi}$,
since $\{\psi_k\}_{n} = \frac{1}{k}\,\delta_{kn}$ and  $\{\phi_k\}_{n} = {k}\,\delta_{kn}$.
 In addition, the upper semi-frame $\Psi$ is totally regular, for the same reason.

Now, in this example, we can identify the end spaces in the scale $\h_{n}:= D(G^{-n/2}), nÊ\in \ZN$. We get
\be\label{eq:schwartz}
\h_{\infty}(G^{-1/2})=\bigcap_{n\in \ZN} \h_n = s, \qquad \h_{-\infty}(G^{-1/2})=\bigcup_{n\in \ZN} \h_{n} = s^\times,
\en
the space of fast decreasing, resp. slowly increasing, sequences (the so-called Hermite representation of tempered distributions 
\cite{simon-herm}). And, indeed, $s$ is a nuclear space, the proof using precisely this representation.

{The example $(\frac{1}{k}e_k), (ke_k)$ can be generalized to weighted  sequences  $(\psi_{k}), (\phi_{k})$, with
$\psi_{k}:=m_k e_k$,  $\phi_{k}:=\frac{1}{m_k} e_k$,  where $m\in\ell^\infty$  has a subsequence converging to zero and $m_k\neq 0, \,\forall \,k$.
Hence the former is an upper semi-frame and not a frame, 
whereas the latter satisfies the lower frame condition, but   not the upper one.
For instance, the sequence $(\frac{1}{2} e_1, \frac{1}{2} e_2, \frac{1}{2^2} e_1, \frac{1}{3} e_3, \frac{1}{2^3} e_1, \frac{1}{4} e_4, \dots )$ is an upper semi-frame of that type.
The frame operator associated to the sequence $(m_k e_k)$ is still diagonal, namely,  $S= \mathrm{diag}(m_{n}^2)$. Thus $S^{-1}= \mathrm{diag}(m_{n}^{-2})$, which is clearly unbounded,
and $\psi_{k}  = S^{-1}\phi_{k}$. The inner products read as:
\begin{itemize}
\vspace{-1mm}\item[{\bf .}] For ${\h}_{\Psi}: \quad\ip{c}{d}_{\Psi} = \sum_{n}{m_{n}^{-2}}\, \ov{c_{n}}\,d_{n}$; 
\vspace{-1mm}\item[{\bf .}]   For ${\h}_0: \quad\;\ip{c}{d}_{0} =\sum_{n}  \ov{c_{n}}\,d_{n}$; 
\vspace{-1mm}\item[{\bf .}]   For ${\h}_{\Psi}^{\times}: \quad\ip{c}{d}_{\Psi}^{\times} =\sum_{n} m_{n}^{2} \,\ov{c_{n}}\,d_{n}$. 
 \end{itemize}
Here too, the upper semi-frame $(\psi_{k})$ is totally regular, since  $\{\psi_k\}_{n} = m_k\,\delta_{kn}$.}
The considerations made above about the triplet \eqref{eq:disctriplet} or the relations \eqref{eq:schwartz} can be made in the case of weighted sequences as well.
For instance, if the sequence $({1}/{m_k})$ grows polynomially, one gets the same result : the end spaces $\h_{\infty}(G^{-1/2})=\bigcap_{n} \h_n$,
 resp. $\h_{-\infty}(G^{-1/2}) =\bigcup_{n} \h_{n}$  still coincide with $s$ and $ s^\times$, respectively. 

Furthermore, if $\Psi$ is an upper semi-frame and there exists weights  $m=(m_n)$ such that   $(\phi_k):=(m_k\psi_k)$ is a frame for $\Hil$, then the following series 
expansions hold true:
\be\label{eq:shifting}
 f = \sum_{k\in\Gamma} \ip {\phi_k}{f} \widetilde\phi_k  = \sum_{k\in\Gamma} \ip {\psi_k}{f}\, \overline{m}_k \,\widetilde\phi_k
=  \sum_{k\in\Gamma} \ip {\overline{m}_k \,\widetilde\phi_k}{f} \psi_k, \ \forall f\in\Hil,
\en
where $(\widetilde\phi_k )$ is the canonical dual of  $(\phi_k )$.
Thus,  the sequence $(\overline{m}_k \,\widetilde\phi_k) $ is   {\xxl a} dual  of $(\psi_k)$. This covers the simple examples above.
This is linked to the invertibility of multipliers \cite{stobal09}, in particular to the following questions: (i) Can the inverse
be represented as in \eqref{eq:shifting}  after shifting weights?  (ii) Is the inverse   a multiplier again? See \cite{ stobal11,stobalrep11}.

The next step would be to consider sequences of the form $\Psi = (\psi_k) :=   (Ve_k)$,  where $V$ is a nice, but nondiagonal operator.
According to \cite[Proposition 4.6]{BSA}, we have the following situations:
\bei
\vspace*{-1mm}\item[(i)] $\Psi$ is a Riesz basis if and only if  $V : \Hil \rightarrow \Hil$ is a bounded bijective operator.
\vspace*{-1mm}\item[(ii)] $\Psi$ is a frame if and only if  $V : \Hil \rightarrow \Hil$ is a bounded surjective operator.
\vspace*{-1mm}\item[(iii)] $\Psi$ is an upper semi-frame if and only if  $V : \Hil \rightarrow \Hil$ is a bounded operator.
\vspace*{-1mm}\item[(iv)] $\Psi$ is a lower semi-frame if and only if $V : \dom(V) \to \Hil$ is a densely defined operator 
such that $e_k\in\dom(V), \,\forall \,k\in\Gamma$, $V^*$ is injective with bounded inverse 
on $\ran(V^*)$, and $V(\sum_{k=1}^n c_k e_k)\to V(\sum_{k=1}^\infty c_k e_k)$ as $n\to\infty$ for every $\sum_{k=1}^\infty c_k e_k\in\dom(V)$. 
\eni
From the discussion of the simple example above, we see that  $(\psi_k) =   (Ve_k)$ is an upper semi-frame  and not a frame if and
only if $V$ is a bounded operator with dense range not equal to $\H$, so that $V^*$ is injective.  This is clearly impossible with any  finite rank operator.
On the other hand, if one takes $V$ to be block-diagonal, with finite dimensional blocks, one is led to fusion
(semi-)frames, as described in Section \ref{sec:rankn-discrframes} (note, however, that the subspaces constituting fusion frames can also be infinite dimensional). 
One might also end up with a controlled (semi-)frame \cite{weightedfr}. 
A whole research field opens up here.

A  useful property of a frame $\Psi = (\psi_k)$ for $\Hil$ is that every element in $\Hil$ can be represented as a series expansion  of the form \eqref{ff}
via some sequence $\Phi = (\phi_k)$. However, there exist Bessel sequences $\Psi$ for $\Hil$  which are not frames and for which \eqref{ff} holds 
via a sequence $\Phi$ , for example, the sequences $\Psi=(\frac{1}{k}e_k)$ and $\Phi=(ke_k)$ discussed above. 
Thus, the frame property is sufficient, but not necessary for series expansions of the form \eqref{ff}.
{As a matter of fact, if one requires a series expansion via a Bessel sequence  which is not a frame, then the dual sequence cannot be Bessel, because of the following result
\cite[Proposition 6.1]{casaz-han} .
\begin{prop.}  \label{prop:nobessel}
If two Bessel sequences $(\psi_k), (\phi_k)$ are dual to each other, then both of them are frames.
\end{prop.}
The simple examples above  give  series expansions \eqref{ff} for all the elements of the space.} The next step is to investigate in general
 what are the possibilities for series expansions via upper semi-frames.

First we note that these simple examples lead us to the general notion  of duality. Indeed, if $\Psi$ is a frame with bounds ({\sf m}, {\sf M}), its canonical dual $\widetilde\Psi$
 is a frame with bounds $({\sf M}^{-1}, {\sf m}^{-1}) $.
Now, formally, an upper semi-frame  $\Psi$ corresponds to  ${\sf m} \rightarrow 0$,  
and yields $S$ bounded,  $S^{-1}$ unbounded. Thus the `dual' $\widetilde\Psi$  should be a sequence satisfying
 the lower frame condition (no finite upper bound, $M \rightarrow \infty$), which would then  correspond to $S$ unbounded and  $S^{-1}$ bounded.
 Actually this idea is basically correct, with some minor qualifications.
Indeed, for an upper semi-frame,   $S:\Hil\to\Hil$ is a bounded injective operator and $S^{-1}$ is unbounded. For a lower semi-frame,   
$S: {\sf Dom}(S)\to\Hil$ is an injective operator, possibly unbounded, with a bounded inverse $S^{-1}$. Indeed, if $\Psi$ is a lower frame sequence for $\Hil$ 
with lower frame bound  {\sf m}   and if $Sf=0$ for some $f\in {\sf Dom}(S)$, then $\ip{f}{Sf} =0$ and thus $\sum_n  |\ip{\psi_n}{f}|^2=0$, which implies that $f=0$, because $\Psi$ is total;
 furthermore, for $f\in {\sf Dom}(S)$ one has ${\sf m}\|f\|^2\leq \ip {f}{Sf}\leq \|Sf\|\cdot\|f\|$, which implies that
 $\|S^{-1}g\|\leq \frac{1}{m}\|g\|$, $\forall g\in {\sf Dom}(S^{-1})$.

Thus there is an almost perfect symmetry (or duality) between two classes of \emph{total}  sequences, namely, those  satisfying the upper frame condition, that is,  upper semi-frames, 
and those  satisfying the lower frame condition, that is,  lower semi-frames. 
For the sake of completeness, we reproduce  two known  results.
\belem   {\rm \cite[Lemma 3.1]{casoleli1}}
Given any total family $\Phi =\{\phi_{k}\}$, the associated analysis operator $C$ is closed. Then $\Phi $ satisfies the lower frame condition,
 i.e. it is a lower semi-frame,  if and only if $C$ has closed range and is injective.
\label{lem:discranalop}
\enlem
Then the main result is the following:
 \begin{prop.} {\rm \cite[Proposition 3.4]{casoleli1}}
Let $\Phi =\{\phi_{k}\}$ be any total family in $\Hil$. Then $\Phi $ is a lower semi-frame if and only if there exists an upper semi-frame
$\Psi $ dual to $\Phi $, in the sense that
$$
f = \sum_{k}  \ip{\phi_{k}}{f} \, \psi_{k}   , \;\; \forall\, f\in \dom(C).
$$
\label{prop:discrduality}  
\end{prop.}
\vspace*{-5mm}
Further results along these  lines may be found in \cite{BSA,casoleli1} to which we refer.
We will explore this symmetry in detail in another work  \cite{jpa-Besseq}.

\section{\sc Generalization of discrete frames}
\label{sec:rankn-discrframes}

\subsection{\sc Fusion (semi-)frames }
\label{subsec:fusionframes}

Rank-$n$ frames were introduced in \cite[Section 2]{jpa-contframes} in the general case of a measure space $(X,\nu)$.
Now, in the purely discrete case, $X$ a discrete set, this concept obviously reduces to an ordinary frame. 
Yet there are plenty of nontrivial generalizations, as soon as one attributes weights to the various subspaces. 

The first step is to consider \emph{weighted frames}, studied in \cite{weightedfr}. Given a set of positive weights $v_{k}>0$, the family 
$\{{\psi}_{k}, k\in \Gamma\}$ is a weighted frame if
$$
{\sf m}  \norm{}{f}^2  \le \sum_{k\in \Gamma}  v_{k}^2\, |\ip{\psi_{k}}{f}| ^2   \le {\sf M}  \norm{}{f}^2 ,  \forall \, f \in \Hil.
$$
Suppose now the weights are constant by blocks of finite size $n_{ j}$, so that one has
$$
{\sf m}  \norm{}{f}^2  \le \sum_{j\in J}  v_{j}^2 \sum_{i=1}^{n_{j}} |\ip{\psi_{ij}}{f}| ^2   \le {\sf M}  \norm{}{f}^2 ,  \forall \, f \in \Hil.
$$
Then, for each $j$, the family  $\{\psi_{ij}, i = 1, 2, \ldots , n_{j} \}$ is a frame for 
 its span, call it $ \Hil_{j}$, which is  at most $n_{j}$-dimensional.  Call $\pi_{\Hil_{j}}$ the corresponding orthogonal projection. 
Let ${\sf m} _j, {\sf M} _j$ be the frame bounds,
$$
{\sf m} _j \norm{}{\pi_{\Hil_{j}} f}^2 \le \sum_{i=1}^{n_{j}} |\ip{\psi_{ij}}{f}| ^2 \le {\sf M} _j \norm{}{\pi_{\Hil_{j}} f}^2,
$$
and   assume  that  ${\sf m}_{\rm inf} := \inf_j {\sf m}_j > 0$ and ${\sf M} _{\rm sup} := \sup_j {\sf M} _j < \infty$.
Then we get 
\be\label{eq:fusionframe}
\frac{\sf m}{{\sf M} _{\rm sup}}  \norm{}{f}^2  \le \sum_{j\in J} v_{j}^2 \,\norm{}{\pi_{\Hil_{j}}f}^2 
 \le \frac{\sf M}{{\sf m}_{\rm inf}}  \norm{}{f}^2 ,  \forall \, f \in \Hil.
\en

In that case, the family $\{\Hil_{j}\}_{ j\in J} $ is a \emph{fusion frame}\footnote{Initially called `frame of subspaces' in  \cite{Casaz-kutyn}.}
 with respect to the weights $\{v_{j} \} _{ j\in J} $, a notion introduced by Casazza and Kutyniok \cite{Casaz-kutyn,Casaz2, sun}.
 Actually, in the general definition, the subspaces
 $\{\Hil_{j}\}_{ j\in J} $ are closed subspaces of $\Hil$, of arbitrary dimension. This structure nicely generalizes frames, in particular, it yields an associated analysis, 
synthesis and frame operator and a dual object.

Given the family $\{\Hil_{j}\}_{ j\in J}  $, one considers their direct sum
$$
\Hil^{\oplus} := \bigoplus_{ j\in J} \Hil_{j} = \{ \{f_{j}\}_{ j\in J }: f_{j}\in \Hil_{j}, \sum_{j\in J}\norm{}{f_{j}}^2 < \infty\}\}\, ,
$$
and this is the ambient Hilbert space. In terms of $\Hil^{\oplus}$, one considers, following the standard pattern,\footnote{In the definitions given in
 \cite[Sec. 3.6]{jpa-bal}, the notations  for the two operators $C_{W,v}$ and $D_{W,v}$ have been interchanged.}

(i) The  {analysis  operator} $C_{W,v} : \Hil \to \Hil^{\oplus}$;

(ii) The  {synthesis operator} $D_{W,v} = {C_{W,v}}^{^{\scriptstyle\!\!\!\!\!\ast}} : \Hil^{\oplus} \to \Hil$;

 (iii)  The  {frame operator} $S_{W,v} : \Hil \to \Hil$ given, as usual, by $S_{W,v}=  {C_{W,v}}^{^{\scriptstyle\!\!\!\!\!\ast}}  \;C_{W,v} $.

Most of the standard results about ordinary frames extend to fusion frames, for instance, the duality relation and the reconstruction formula. 
We refer to \cite{jpa-bal} or the original papers for details.

In \cite[Theorem 3.2]{Casaz-kutyn} and \cite[Theorem 2.3]{Casaz2} , the authors establish an equivalence between frames and fusion frames, under some mild conditions.
We will now extend this result to semi-frames, following the same scheme. By analogy with the frame case, we define an \emph{upper fusion semi-frame} in $\H$ as a family of closed subspaces
$\{ \Hil_{j}, j\in J\}$ for which the following relation holds, with some weights $v_j \neq 0$ and an upper bound ${\sf M}<\infty$:
$$
0 <  \sum_{j\in J} v_{j}^2 \,\norm{}{\pi_{\Hil_{j}}f}^2  \leq  {\sf M}   \norm{}{f}^2 ,  \forall \, f \in \Hil.
$$
First we consider a  family  of  closed subspaces of $\H$ and we build an upper fusion semi-frame out of them.
Conversely, given an upper fusion semi-frame, we build an upper semi-frame.
\begin{prop.}
(1) Given two index sets $J$ and $I_j$, finite or not, let the family  $\{\psi_{ij}, j\in J, i \in I_{j} \}$ be an upper semi-frame in $\H$ 
with bound ${\sf M}$.
For each $j\in J$, denote by $ \Hil_{j}$ the   closure of ${\sf Span} \{\psi_{ij}, i\in I_j\}$ and
assume that $\{\psi_{ij}, i \in I_{j}\}$ is a lower semi-frame in $\Hil_{j}$ with lower bound ${\sf m} _j$. 
Then  the family $\{ \Hil_{j}, j\in J\}$, is an upper fusion semi-frame in $\H$, with upper bound $ {\sf M}$ and weights  ${\sf m} _j$.

(2) Conversely, for every $j\in J$, let  $\{\psi_{ij},  i \in I_{j} \}$ be an upper semi-frame for the closure  $\H_j$ of their span, with upper bound
 ${\sf M} _j$. Assume that ${\sf M}  := \sup_j {\sf M} _j < \infty$ and that the family $\{ \Hil_{j}, j\in J\}$, is an upper fusion semi-frame in $\H$
 with weights $v_j$ and bound $\sf B$.
 Then $\{\psi_{ij}, j\in J, i \in I_{j} \}$ is a weighted upper semi-frame for $\H$ with weights $v_j$ and bound {\sf M}{\sf B}.
\end{prop.}
\begin{proof}
The proof of (1) is similar to that of  \cite[Theorem 3.2]{Casaz-kutyn}. Indeed, by assumption,
$$
{\sf m} _{j }\norm{}{\pi_{\Hil_{j}}f}^2 \leq \sum_{i\in I_j }|\ip{\psi_{ij}}{f}|^2.
$$
Therefore,
$$
 \sum_{j\in J } {\sf m} _j \norm{}{\pi_{\Hil_{j}}f}^2 \leq \sum_{j\in J } \sum_{i\in I_j } |\ip{\psi_{ij}}{f}|^2 \leq  {\sf M} \norm{}{f}^2.
$$
As for (2),we have
\begin{align*}
\sum_{j\in J } \sum_{i\in I_j } v_j^2 |\ip{\psi_{ij}}{f}|^2 
&= \sum_{j\in J }v_j^2 \sum_{i\in I_j }|\ip{\psi_{ij}}{\pi_{\Hil_{j}}f}|^2 \\
&\leq  \sum_{j\in J }v_j^2 \, {\sf M} _j  \norm{}{\pi_{\Hil_{j}}f}^2
\leq {\sf M} \sum_{j\in J }v_j^2 \, \norm{}{\pi_{\Hil_{j}}f}^2 \leq {\sf M}{\sf B}\norm{}{f}^2.
\end{align*}
It is easy to see that the totality condition is satisfied in both cases. 
\end{proof}
\medskip

\subsection{\sc Semi-frames in Banach spaces and beyond}
\label{subsec:banach}

Frames in Banach spaces have been defined and studied by several authors, see for instance \cite{aldroubi,casaz-han,casole-st,sto09}. As we shall see, 
most of their results can be extended to semi-frames as well. Throughout  this section, $X_{d}$ denotes a Banach space of sequences $c=(c_{k}), \, k\in\Gamma$, such that the coordinate
 linear functionals $c\mapsto c_{k}$ are continuous on $X_{d}$ and $\B$ is a separable Banach space, with dual $\B\ha$.

{Since there is no inner product on a general Banach space, a frame is defined as an indexed set of linear functionals $(\mu_{k})$ from $\B\ha$. Thus one defines:
\bei
\vspace*{-1mm}\item[(a)] The family $ (\mu_{k}) \subset \B\ha$ is an \emph{$X_{d}$-frame} for $\B$ if
\bei
\vspace*{-1mm}\item[(i)]   $(\mu_{k}(f))\in X_{d}, \, \forall \,f\in \B$;
\vspace*{-1mm}\item[(ii)] the norms  $\norm{\B}{f}$ and $\norm{X_{d}}{(\mu_{k}(f))}$ are equivalent, \emph{i.e.},  
there exist constants  ${\sf m}>0$ and ${\sf M}<\infty$ such that
 \be
{\rm\sf m} \norm{\B}{f} \le   \norm{X_{d}}{(\mu_{k}(f))}  \le {\rm\sf M}  \norm{\B}{f} ,  \forall \, f \in \B.
\label{eq:Bframe}
\end{equation}
\eni
\vspace*{-1mm}\item[(b)]
The family $ (\mu_{k})$ is an \emph{$X_{d}$-Banach frame} if, in addition, there exists a bounded linear operator
$S: X_{d}\to X$ such that $S (\mu_{k}(f)) = f, \, \forall\, f\in \B$. However, this does not imply the existence of a reconstruction formula in terms of an infinite series.
\vspace*{-1mm}\item[(c)]
 The family $ (\mu_{k})$ is an \emph{$X_{d}$-Bessel sequence} if only the upper inequality in \eqref{eq:Bframe} is satisfied and 
an \emph{$X_{d}$-upper semi-frame} if, in addition,  it is total in $\B\ha$ (that is, $\mu_{k}(f)=0, \forall\, k,$ implies $f=0$). It is an \emph{$X_{d}$-lower semi-frame}
if only the lower inequality in \eqref{eq:Bframe} is satisfied.
\eni
In case $X_{d}= \ell^p$, one speaks of $p$-(semi-)frames, etc \cite{aldroubi}. In fact, a 2-frame can always be reduced to a Hilbert frame.}

The new fact here, as compared to the Hilbert case, is that the properties of the (semi-)frames depend crucially on those of the sequence space $X_{d}$.
A systematic analysis has been made by Casazza \emph{et al.} \cite{casole-st} and Stoeva  \cite{sto09}, to which we refer for more details. 
The crucial property is the following. The space $X_{d}$ is called a CB-space if the canonical unit vectors $\{e_{k}\}$ form a Schauder basis of it.
In that case, the dual  $X_{d}\ha$ may be identified with a sequence space via the isomorphism $h\in X_{d}\ha \leftrightarrow  (h(e_{k})) $.

Using this language, we may quote some results about semi-frames. To that effect, given a sequence $ (\mu_{k})\subset \B\ha$, 
we introduce an associated analysis operator   $C: \B\to X_{d} $ by 
$ C f = (\mu_{k}(f)) $   on the domain  $\dom(C)= \{f\in\B :  (\mu_{k}(f)) \in X_{d}  \}$. 
In general, the subspace $\dom(C)$ need not be closed.
If $ (\mu_{k})$ is an $X_{d}$-frame, however, $\dom(C)= \B$ and $C$ is bounded and an isomorphism. 

Concerning $X_{d}$-upper semi-frames, we have the following results,  reminiscent from Lemma \ref{lem31}(1):

\begin{prop.} 
(a) {\rm \cite[Proposition 3.2]{casole-st}} 
Let $X_{d}$ be a CB-space. Then  $ (\mu_{k})\subset \B\ha$ is an $X_{d}\ha$-Bessel sequence with bound {\sf M}
if and only if the synthesis operator $D': ( d_{k}) \mapsto \sum_{k}d_{k}\mu_{k} $ is bounded from $X_{d}$ into $\B\ha$ and  $\norm{}{D'} \le {\sf M}$. 

(b){\rm \cite[Corollary 3.3]{casole-st}} and  {\rm \cite[Proposition 3.2]{sto09}} 
Assume that    $X_{d}\ha$ is a  CB-space. If $ (\mu_{k})\subset \B\ha$ is an $X_{d}$-Bessel sequence with bound {\sf M}, then 
the synthesis operator $D:  (d_{k})  \mapsto \sum_{k}d_{k}\mu_{k} $ is bounded from $X_{d}\ha$ into $\B\ha$ and
  $\norm{}{D} \le {\sf M}$. The converse is true if $X_{d}$ is reflexive.
\end{prop.}

\begin{cor.}{\rm \cite[Lemma 3.3]{sto09}}
 If $X_{d}$ and $X_{d}\ha$ are both  CB-spaces, $C\ha = D$ and $C = D\ha{|_{\B}} $.
\end{cor.}
As for the last statement, note that $\B \subset \B\haa$ in general, unless $\B$ is reflexive.

Concerning $X_{d}$-lower semi-frames, we have the following result, analogous to Lemma \ref{lem:discranalop}
\belem {\rm \cite[Lemma 3.5]{sto09}}
Let $ (\mu_{k})\subset \B\ha$ satisfy the following restricted lower bound condition:   there exists a constant ${\sf m}>0$ such that
\be
{\sf m}  \norm{\B}{f} \le  \norm{X_{d}}{(\mu_{k}(f))}\, , \;\; \forall \, f \in \dom(C).
\label{eq:disclowersf2} 
\end{equation}
Then the analysis operator $C$ is injective, closed and its range $R_{C}$ is closed in $X_{d}$. The inverse  $C^{-1}: R_{C} \to \dom(C)$ 
is bounded with norm $\norm{}{C^{-1}} \le 1/{\sf m}$.
\enlem
In the case of an $X_{d}$-lower semi-frame, $\dom(C) = \B$ by definition and the result applies. 

Finally, the following theorem  summarizes the case of lower semi-frames.
\begin{theorem}{\rm \cite[Theorem 3.8]{sto09}}
Given a CB-space $X_{d}$   and a sequence $ (\mu_{k})\subset \B\ha$, consider the following conditions:
\bei
\item[$\sf{(b_1)}$] There exists an   $X_{d}\ha$-Bessel sequence $(g_{k})\subset \B$ for $\B\ha$ such that
$$
f = \sum_k \mu_{k}(f)\,g_{k}, \, \forall\,f\in \dom(C);
\vspace*{-2mm}
$$
\item[$\sf{(b_2)}$] The operator $C^{-1} : R_C \to \B$ has a bounded extension to $X_{d}$.
\eni
Then $\sf{(b_1)}$ holds true if and only if $ (\mu_{k})$ satisfies the relation \eqref{eq:disclowersf2} and $\sf{(b_2)}$ holds true.
\end{theorem}
Further results, in particular about reconstruction formulas via series expansions, may be found in \cite{casole-st} and   \cite{sto09}.}
\medskip

An interesting alternative  recently proposed in \cite{zhang}, using the notion of semi-inner product, allows us to define (in certain cases) 
$X_{d} $-(semi-)frames as a family of vectors $ (\psi_k)$ from $\B$ instead of $\B\ha$.

A \emph{semi-inner product} \cite{jpa-sips,lumer} on $\B$ is a function $[\cdot,\cdot]$ on $\B \times \B$ such that, for every $f,g,h\in\B$ and $\alpha\in\CN$,
\bei
\vspace*{-1mm}\item $[f,g+h] = [f,g] + [f,h]$,
\vspace*{-1mm}\item $[f,\alpha g ] = \alpha[f,g] $ and $[\alpha  f,g ] = \ov{\alpha}[f,g] $,
\vspace*{-1mm}\item $[f,f ]>0$ for $f\neq 0$,
\vspace*{-1mm}\item $ |[f,g ]|^2 \leq [f,f ]\,[g,g ].$
 \eni
Note, $[\cdot,\cdot]$ is \emph{not} additive in the first factor, lest it becomes a genuine inner product. Actually, every Banach space $\B$ has a semi-inner product 
$[\cdot,\cdot]$ that is compatible, \emph{i.e.}, $[f,f ]^{1/2} = \norm{\B}{f},\, \forall\, f\in\B$.

Using this terminology, one calls  a sequence $ (\psi_k) $ an $X_{d}$-frame for $\B$ if $([\psi_k,f] )\in X_{d}, \forall\, f\in\B$ and there exist  
 constants  ${\sf m}>0$ and ${\sf M}<\infty$ such that
 \be
{\rm\sf m} \norm{\B}{f} \le   \norm{X_{d}}{([\psi_k,f])}  \le {\rm\sf M}  \norm{\B}{f} ,  \forall \, f \in \B,
\label{eq:sipframe}
\end{equation}
and similarly for the other notions. Further steps require additional restrictions on $\B$, namely, that $\B$ be reflexive and strictly convex, that is,
 $\norm{\B}{f+g}= \norm{\B}{f} + \norm{\B}{g}$ for $f,g\neq 0$ implies that $f=\alpha g$ for some $\alpha >0$. In that case,
the duality mapping from $\B$ to $\B\ha$ is bijective, that is, for every functional $\mu\in\B\ha$, there exists a unique $f\in\B$ such that 
$\mu(g) = [f,g], \, \forall\, f\in \B$. This allows us to define in a unified way the various structures such as $X_{d}$-frames, $X_{d}\ha$-frames, etc. 
For a thorough analysis, we refer to \cite{zhang}.

As a last generalization, we note that  Pilipovi\'c \emph{et al.} \cite{pil-sto07,pil-sto08} have extended the construction of frames to 
Fr\'echet spaces, more precisely, to a projective limit of reflexive Banach spaces 
$\B_\infty := \bigcap_{s\in \ZN} B_s$, where, for $s\in \NN_0,  B_{-s} := B_s\ha$ and
$$
\B_\infty \subset \ldots \subset B_2 \subset B_1 \subset B_0 \subset B_{-1}\subset B_{-2}\subset \ldots .
$$
Actually, such a Banach scale is a simple example of partial inner product space \cite{jpa-pipspaces}. Thus this construction might still be generalized considerably.
This will be the subject of future work.

\section{\sc Continuous frames and semi-frames}
\label{sec:contframes}

\subsection{\sc Continuous frames revisited}
\label{sec:contsemiframes}

We turn now to the   continuous  generalized frames, introduced in Section \ref{sec:intro},  eq. \eqref{eq:genframe}. 
A complete analysis has been made in our previous paper \cite{jpa-bal}, so we will be rather brief here.

Let $\Hil$ be a  Hilbert space and   $X$ a  locally compact, $\sigma$-compact, space with measure $\nu$. 
Let   $\Psi:=\{\psi_{x},\, x\in X\}, \,\psi_{x}\in \Hil $ be a continuous frame, as defined in \eqref{eq:frame}.

First,  $\Psi$ is a total set in $\Hil$. Next  define the  \emph{analysis operator} by the (coherent state) map $C_{\Psi}: \Hil \to L^{2}(X, \ud\nu)$ given as 
$$
(C_{\Psi}f)(x) =\ip{\psi_{x}}{f} , \; f \in \Hil,
$$
 with range by $R_{C}:=   \ran(C_{\Psi})$.
 Its adjoint $C_{\Psi}^\ast: L^{2}(X, \ud\nu) \to \Hil$, called the \emph{synthesis operator},
 reads (the integral being understood in the weak sense,  as usual \cite{forn-rauhut})
\be
C_{\Psi}^\ast F =  \int_X  F(x) \,\psi_{x} \; \ud\nu(x), \mbox{ for} \;\;F\in L^{2}(X, \ud\nu)  
\label{eq:synthmap}
\end{equation}
Then $C_{\Psi}^* C_{\Psi} = S$ and $ \|C_{\Psi} f\|^2_{L^{2}(X)}= \| S^{1/2}f\|_{\Hil}^2 = \ip{f}{Sf}$.
 Furthermore,  $C_{\Psi}$ is injective, since $S>0$,  so that $ C_{\Psi}^{-1} : R_{C}  \to \Hil$ is well-defined.

Next,   the lower frame bound implies that $R_{C}$  is a \emph{closed} subspace of $L^{2}(X, \ud\nu )$. 
The  corresponding orthogonal projection is  ${\mathbb P}_{\Psi}:  L^{2}(X, \ud\nu) \to R_C$ defined by
$$
{\mathbb P}_{\Psi}:= C_{\Psi} S^{-1} C_{\Psi}^* = C_{\Psi}  C_{\Psi}^+ ,
$$
where $ C_{\Psi}^+ = S^{-1} C_{\Psi}^*$  is the pseudo-inverse of  $C_{\Psi}$.
The projection ${\mathbb P}_{\Psi}$ is an integral operator with (reproducing) kernel $K(x,y) = \ip{\psi_{x} }{S^{-1}\psi_{y}}$,
 thus $R_C$ is a  \emph{reproducing kernel Hilbert space}.  

In addition, the subspace $ R_{C}$    is also  complete in the norm $\| \cdot\|_{\Psi}$, associated to the  inner product 
 \be
\ip {F}{F'}_{\Psi} := \ip {F}{C_{\Psi} \,S^{-1} \,C_{\Psi}^{-1} F' }_{L^{2}(X)} , \mbox{ for } \; F,F'\in R_{C} .
\label{eq:scalarprod}
\en
Hence  $(R_{C},\| \cdot\|_{\Psi}) $ is a Hilbert space, denoted by ${\Hil}_{\Psi}$, and the map    $C_{\Psi}: {\Hil} \to {\Hil}_{\Psi}$ is unitary.
  Therefore, it can be inverted  on its range by the adjoint operator
 $C_{\Psi}^{\ast (\Psi)}: {\Hil}_{\Psi} \to {\Hil} $,  which is precisely the pseudo-inverse  $C_{\Psi}^+ = S^{-1} C_{\Psi}^*.$
Thus one gets,  for every $f\in \Hil$, a  \emph{reconstruction formula}, with a weakly convergent integral:  
\be
f = C_{\Psi}^{\ast (\Psi)} 
=  \int_X  F(x) \,S^{-1}\,\psi_{x} \; \ud\nu(x), \mbox{ for} \;\;F= C_{\Psi}f \in {\Hil}_{\Psi} .
\label{eq:reconstr}
 \en
We should also note that frame multipliers for continuous frames have been developed recently \cite{bayer}. It remains to be seen how much of this can be extended
to (upper) semi-frames.

\subsection{\sc Continuous upper semi-frames}
\label{sec:contuppersemiframes}

Let  now $\Psi$ be     a (continuous) \emph{upper semi-frame}, that is, there exists ${\rm\sf M}<\infty$ such that
\be\label{eq:upframe}
0 < \int_{X}  |\ip{\psi_{x}}{f}| ^2 \, \ud \nu(x)   \leq {\rm\sf M}  \norm{}{f}^2 , \; \forall \, f \in \Hil, \, f\neq 0 .
\en
 In this case,  $\Psi$ is a total set in $\Hil$, the operators $C_{\Psi}$ and $S$ are bounded, $S$ is injective and self-adjoint.
Therefore $R_S:=   \ran(S)$ is dense in $\Hil$ and $S^{-1}$ is also self-adjoint. 
$S^{-1}$ is unbounded, with dense domain $\dom(S^{-1}) = R_S$. 
 
Define the Hilbert space ${\h}_{\Psi}$ as the   completion of $C_{\Psi}(R_{S})$ with respect to the norm $\| \cdot\|_{\Psi}$    defined in \eqref{eq:scalarprod}.
 Then,  the map $C_{\Psi}$ is an  {isometry} from $\dom(S^{-1}) =   R_{S}$ onto  
 $C_{\Psi}(R_{S}) \subset {\h}_{\Psi}$, thus 
it extends by continuity to a  \emph{unitary} map from $\Hil$ onto ${\Hil}_{\Psi}$.
Therefore, $ {\h}_{\Psi} $ and $  R_{C}$ coincide as sets,
 so that  $ {\h}_{\Psi} $ is  a vector subspace (though not necessarily  closed) of  $ L^{2}(X, \ud\nu)$. 
 
Consider now the operators $G_S = C_{\Psi} \,S \,C_{\Psi}^{-1} : R_{C} \to C_{\Psi}(R_{S})$ and
$G_{S}^{-1}:=C_{\Psi} \,S^{-1} \,C_{\Psi}^{-1} : C_{\Psi}(R_{S}) \to R_{C}$, both acting in the Hilbert space $ \ov{R_C}$.
Then one shows \cite{jpa-sqintegI} that $G_{S} $ is a bounded, positive and symmetric operator, while $G_{S} ^{-1}$ is positive and essentially self-adjoint. 
These two operators are bijective and inverse to each other. Thus one gets the following commutative diagram.
\be\label{diagram2}
\hspace*{-1cm}\begin{minipage}{15cm}
\begin{center}
\setlength{\unitlength}{2pt}
\begin{picture}(100,60)
\put(10,50){$\Hil$}
\put(89,50){${\h}_{\Psi}= R_{C} \subseteq  \ov{R_C}\subseteq L^{2}(X, \ud\nu) $}
\put(20,52){\vector(4,0){65}}
\put(55,54){$C_{\Psi}$}
\put(10,47){\vector(0,-4){30}}
Ê\put(12,17){\vector(0,4){30}}
\put(100,47){\vector(0,-4){30}}
Ê\put(102,17){\vector(0,4){30}}
\put(-35,10){$ \Hil \supseteq \dom(S^{-1})=R_S$}
\put(85,47){\vector(-2,-1){65}}
\put(14,30){$S^{-1}$}
\put(5,30){$S$}
\put(91,30){$G_{S}$}
\put(104,30){$G_{S} ^{-1}$}
\put(93,10){$C_{\Psi}(R_{S}) \subseteq L^{2}(X, \ud\nu)$}
\put(20,12){\vector(4,0){65}}
\put(48,35){$C_{\Psi}^\ast$}
\put(50,5){$C_{\Psi}$}
\end{picture} 
\end{center}
\end{minipage}
\end{equation}
 
Next let   $G = \ov{G_{S}}$ and let $G^{-1}$ be the self-adjoint extension of  $G_{S} ^{-1}$. Both operators are self-adjoint and positive, 
$G$ is bounded and  $G^{-1}$ is densely defined   in $\ov{R_C}$. Furthermore, they are are inverse of each other on the appropriate domains.
 Moreover, since the spectrum of  $G^{-1}$ is bounded away from zero,
the norm $\norm{\Psi}{\cdot}$ is equivalent to the graph norm of {$G^{-1/2} = {\left( G^{-1} \right)}^{1/2}$}, so that
$$
 \ran(G^{1/2}) = \dom(G^{-1/2}) = {\h}_{\Psi} = R_{C} \subset \ov{R_{C}}\subset  L^{2}(X, \ud\nu).
$$

As in the discrete case, we will say that the upper semi-frame  $\Psi=\{\psi_{x},\, x\in X\}$ is \emph{regular}
 if  all the vectors $\psi_{x}, \, x \in X$, belong to $\dom(S^{-1})$. 
In that case,  the discussion proceeds exactly as in the bounded case. In particular,  
the reproducing kernel $K(x,y) = \ip{\psi_{x} }{S^{-1}\psi_{y}}$ is a \emph{bona fide} function on $X\times X$.
One obtains the same weak reconstruction formula,  but restricted to the subspace $R_S = \dom(S^{-1}) $: 
\be
f =   C_{\Psi}^{\ast (\Psi)} F =  \int_X  F(x) \,S^{-1}\,\psi_{x} \; \ud\nu(x),   
\forall\;f\in  R_S ,\;  F= C_{\Psi} f \in {\Hil}_{\Psi}. 
\en
 On the other hand, if  $\Psi$ is not regular, one has to treat the kernel $K(x,y)$ as a bounded sesquilinear  form
 over $ {\Hil}_{\Psi}$ and use the language of distributions, for instance, with a Gel'fand triplet \cite{jpa-bal}.

The construction, originating from \cite[Section 3]{jpa-contframes} and \cite[Section 7.3]{jpa-CSbook}, proceeds exactly 
as in the discrete case of Section \ref{subsec:discG-triplet}.
If  $\Psi$ is regular, 
one has indeed
\be\label{eq:sesqform1}
\iint_{X\times X} \ov{F(x)}K(x,y)F'(y)\; \ud\nu(x) \; \ud\nu(y) = \ip{C_{\Psi}^{-1} F}{SC_{\Psi}^{-1}F'}_\Hil, \; \forall\, F,F'\in {\h}_{\Psi}.
\end{equation}
Since $C_{\Psi}$ is an isometry and $S$ is bounded, the relation \eqref{eq:sesqform1} defines a bounded sesquilinear form over ${\h}_{\Psi}$, namely
\be\label{eq:sesqform2}
K^\Psi(F,F') = \ip{C_{\Psi}^{-1} F}{SC_{\Psi}^{-1} F'}_\Hil,
\end{equation}
and this remains  true even  if  $\Psi$ is not regular.  
Denote by ${\h}_{\Psi}^{\times}$ the Hilbert space obtained by completing  ${\h}_{\Psi}$ in the norm given by this sesquilinear form.
Now, \eqref{eq:sesqform1} and \eqref{eq:sesqform2}  imply that 
$$ 
 K^\Psi(F,F') = \ip{C_{\Psi}^{-1} F}{SC_{\Psi}^{-1} F'}_\Hil = \ip{F}{C_{\Psi} SC_{\Psi}^{-1} F'}_\Psi = \ip{F}{F'}_{L^2} .
$$
Therefore, one obtains, with continuous and dense range embeddings,
\be
 {\h}_{\Psi} \;\subset\; {\h}_{0}\;\subset \; {\h}_{\Psi}^{\times},
\label{eq:conttriplet}
\end{equation}
where 
\begin{itemize}
\vspace{-1mm}\item[{\bf .}]  $ {\h}_{\Psi} = R_{C}$, which is a Hilbert space for the norm 
$\norm{\Psi}{\cdot}=\ip {\cdot}{ G^{-1}  \cdot }^{1/2}_{L^2}$; 
\vspace{-1mm}\item[{\bf .}]   ${\h}_{0} =\ov{{\h}_{\Psi}} = \ov{R_{C}}$ is the closure of ${\h}_{\Psi}$ in $L^{2}(X, \ud\nu)$;
\vspace{-1mm}\item[{\bf .}]  
${\h}_{\Psi}^{\times}$ is  the completion of  ${\h}_{0}$  (or ${\h}_{\Psi}$) in the norm $\norm{\Psi^{\times}}{\cdot}:=\ip {\cdot}{G \cdot }^{1/2}_{L^2}$,
as well as the conjugate dual of ${\h}_{\Psi}$.
\end{itemize}
The rest is as in the discrete case.

In particular, \eqref{eq:conttriplet} is the central triplet of the scale of Hilbert spaces generated by the powers of $G^{-1/2}$,
namely,  $\h_{n}:= D(G^{-n/2}), nÊ\in \ZN$. Here again, one may ask the nature and properties of the end spaces of the scale, $\h_{\pm\infty}(G^{-1/2})$, 
defined in \eqref{scale}. 

The question can be made more precise in the case of the non-regular upper semi-frame of coherent states described in \cite[Section  2.6]{jpa-bal}.
The Hilbert space is $\Hil^{(n)}:=L^{2}({\RN}^{+}, r^{n-1}\ud r), n = \hbox{integer} \geq 1 $.\footnote{
There is some confusion in \cite[Section  2.6]{jpa-bal}, as well as in \cite[Section  5]{jpa-contframes},
namely, both papers use the Hilbert space $\Hil^{(n)}:=L^{2}({\RN}^{+}, r^{n+1}\ud r)$, instead of the present one, but yet the correct function $s(r) = \pi r^{n-1}|\psi (r)|^{2}$ 
and operators $S^{\pm 1}$. The error propagates in Eqs. (2.23) and (2.24) of \cite[Section  2.6]{jpa-bal} and in the various norms and 
expressions for $G^{\pm 1}$, each of them containing an extra factor $r^2$.}
The  vector  $\psi_{x}$ is admissible if it satisfies  the two conditions
 \begin{align*} 
  (i) \;  & \sup_{r \in {\RN}^{+}}{\mathfrak s}(r): = 1 , \mbox {where } {\mathfrak s}(r):=\pi r^{n-1}|\psi (r)|^{2}  \label{eq518}
\\
 (ii) \; & |\psi (r)|^{2} \neq 0, \; \hbox{except perhaps  at isolated points} \; r \in {\RN}^{+}\nn.   
\end{align*}
The frame operator $S$ and its inverse $S^{-1}$ are multiplication operators on $\Hil^{(n)}$, namely 
$$
(S^{\pm1}f)(r) =  [{\mathfrak s}(r)]^{\pm1}f(r).
$$
Since ${\mathfrak s}(r)\leq 1$, the inverse $S^{-1}$ is indeed unbounded and no frame vector $\psi_{x}$ belongs to its domain.
Of course, we have also, for every $m\in \ZN$,
$$
(S^{m}f)(r) =  [{\mathfrak s}(r)]^{m}f(r).
$$
Thus the scale generated by $S^{-1/2}$ consists of the spaces $\Hil_m = \dom(S^{-m/2}), m\in \ZN$, with norm
\be\label{eq:normHm}
 \norm{m}{f}\!\!\!\!^{\widetilde \;\;} = \ip{f}{S^{-m} f}^{1/2} = \left[\int_0^\infty |f(r)|^2\, [{\mathfrak s}(r)]^{-m}  \, {r^{n-1} \ud r }\right]^{1/2} .
\en
However, the end spaces of the scale, namely 
$$
 \Hil_{\infty}(S^{-1/2}):=\bigcap_{m\in \ZN} \Hil_m , \qquad \ \Hil_{-\infty}(S^{-1/2}):=\bigcup_{m\in \ZN} \widetilde \Hil_{m} ,
$$
do not seem to have an easy interpretation.

Let us give an example, in the case $n=1$, for simplicity. First, the semi-frame $\psi_{x}$ is not regular. Indeed one has $|\psi_{x}(r)|^2 = \pi^{-1} {\mathfrak s}(r)$, 
so that the norm \eqref{eq:normHm} reads
$$
 \norm{m}{\psi_{x}}\!\!\!\!^{\widetilde \;\;}
 = \left[\pi^{-1}\int_0^\infty {\mathfrak s}(r)\,  {[{\mathfrak s}(r)]^{-m} \, \ud r }\right]^{1/2}   r   =\infty, \; \forall\, m\geq 1.
$$
Next, take $\psi_x$ such that ${\mathfrak s}(r)=  e^{-\alpha r}, \alpha >0$. Then the norm \eqref{eq:normHm} becomes
$$
\norm{m}{f}\!\!\!\!^{\widetilde \;\;}
=\left[ \pi^{-1}\int_0^\infty  |f(r)|^2\,{e^{\alpha m r} \, \ud r, }\right]^{1/2} \; m=0,1,2,\ldots .  
$$
Thus $f\in\Hil_{\infty}(S^{-1/2})$ if $\norm{m}{f}\!\!\!\!^{\widetilde \;\;}<\infty$, for all $  m=0,1,2,\ldots$.
This condition defines a specific type of function space, but have been unable to identify it explicitly. 

In the same way, one has 
$$
(G^{\pm1}F)(x) = \int_{{\RN}^{+}} e^{ixr}\, \widehat F(r)\, {[{\mathfrak s}(r)]^{\pm1} \, ud r, } 
$$
and, for every $m\in \ZN$,
$$
(G^{m}F)(x) = \int_{{\RN}^{+}} e^{ixr}\, \widehat F(r)\,{[{\mathfrak s}(r)]^{m} \, \ud r. }  
$$
Accordingly, the associated Hilbert scale consists of the spaces $\h_m = \dom(G^{-m/2}), m\in \ZN$, with norm
$$
 \norm{m}{F}  = \ip{F}{G^{-m} F}^{1/2} = \left[2\pi\int_0^\infty |\widehat F(r)|^2\, {[{\mathfrak s}(r)]^{-m} \,\ud r } 
 \right]^{1/2}.
$$
Here too, the end spaces 
$$
\h_{\infty}(G^{-1/2}):=\bigcap_{m\in \ZN}   \h_m , \qquad   \h_{-\infty}(G^{-1/2}):=\bigcup_{m\in \ZN}   \h_{m} ,
$$
do not seem   easy to identify.

\subsection{\sc Lower semi-frames, duality}
\label{sec:duality}

Given a frame $\Psi = \{\psi_{x}\}$,  one says \cite{zakharova} that a frame $\{\chi_{x}\}$ is \emph{dual} to the frame $\{\psi_{x}\}$ if one has, in the weak sense,
 $f = \int_X  \ip{\chi_{x}}{f} \,\psi_{x} \; \ud\nu(x), \; \forall\,f\in\Hil $.
Then the frame $\{\psi_{x}\}$ is dual to the frame  $\{\chi_{x}\}$. This applies, in particular, to a given frame  $\Psi = \{\psi_{x}\}$ and 
its canonical  dual $\widetilde\Psi = \{\widetilde\psi_{x}:= S^{-1}\psi_{x}\}$
We want to extend this notion to semi-frames.
 It is known  \cite{gabardo-han} that an  upper semi-frame  $\Psi$ is a frame if and only if there exists another  upper semi-frame  $\Phi$ which is dual to $\Psi$, in the sense that
$$
\ip{f}{f'} = \int_X  \ip{\phi_{x}}{f} \,\ip{\psi_{x}}{f'} \; \ud\nu(x), \; \forall\, f,f' \in \Hil. 
$$

Let first $\Psi = \{\psi_{x}\}$ be an arbitrary total family in $\Hil$.
Then we define the analysis operator $C_{\Psi}: \dom(C_{\Psi})\to L^{2}(X, \ud\nu)$ as $C_{\Psi}f(x) = \ip{\psi_{x}}{f}$ on the domain
$$
\dom(C_{\Psi}):= \{f\in\Hil : \int_{X}  |\ip{\psi_{x}}{f}| ^2 \, \ud \nu(x) <\infty\}.
$$
Next,    we define the synthesis operator $D_{\Psi}: \dom(D_{\Psi})\to\Hil $ as 
\be
D_{\Psi}F =  \int_X  F(x) \,\psi_{x} \; \ud\nu(x),  \;\;F\in \dom(D_{\Psi}) \subset L^{2}(X, \ud\nu) ,
\label{eq:synthmap2}
\end{equation}
\textcolor{black} {on the domain 
$$
\dom(D_{\Psi})  := \{F\in L^{2}(X, \ud\nu)  : \int_{X} F(x) \,\psi_{x}  \, \ud \nu(x)  \mbox{ converges weakly in $\Hil$ }\}.  
$$
\emph{A priori} they are both unbounded.
Following \cite[Lemma 3.1]{casoleli1} and  \cite[Lemma 3.1 and Proposition 3.3]{BSA},  we have, as in the discrete case}
\begin{lem.}\label{lem:analop}
(i) Given any total family $\Psi$,  
the  analysis operator $C_{\Psi}$ is closed. Then $\Psi $ satisfies the lower frame condition
 if and only if $C_{\Psi}$ has closed range and is injective.

(ii) If  the function $x\mapsto \ip{\psi_{x}}{f}$ is locally integrable for all $f\in\Hil$, then the operator 
$D_{\Psi}$ is densely defined and one has $C_{\Psi} = D_{\Psi}^\ast$.
\end{lem.}
A proof is given in \cite[Lemmas 2.1 and 2.2]{jpa-bal}.
The condition of local integrability is   satisfied for all  $f\in\dom(C_{\Psi})$, but not necessarily for all $f\in\Hil$, unless $\Psi$ is 
an upper semi-frame, since then $\dom(C_{\Psi})=\Hil$.

Finally, one defines the frame operator as $S=D_{\Psi}  C_{\Psi}$, so that, in the weak sense,
$$
 {Sf} =  \int_{X}  \ip{\psi_{x}}{f} \psi_{x} \, \ud \nu(x) , \;\; \forall \,f\in \dom(S),  
$$
where
$$
\dom(S)  := \{f\in \Hil  : \int_{X} \ip{\psi_{x}}{f}\, \psi_{x}   \, \ud \nu(x)  \mbox{ converges weakly in $\Hil$ }\}.
$$
Notice that one has in general $\dom(S) \subsetneqq \dom(C_{\Psi})$. As in the discrete case  \cite[Lemma 3.1]{BSA}, one has
$\dom(S) = \dom(C_{\Psi})$ if and only if $\ran(C_{\Psi})\subseteq \dom(D_{\Psi}) $. This happens, in particular, for an upper semi-frame $\Psi$, for which one has
$\dom(S) = \dom(C_{\Psi}) = \Hil$.

For an upper  semi-frame,   $S:\Hil\to\Hil$ is a bounded injective operator and $S^{-1}$ is unbounded. If $\Phi=\{\phi_{x}\} $ satisfies the lower frame condition,  
 then $S: \dom(S)\to\Hil$ is an injective operator, possibly unbounded, with a bounded inverse $S^{-1}$. 
However, if the upper frame inequality is not satisfied, $S$ and $C_{\Psi}$ could have nondense domains,  
 in which case one cannot define a unique adjoint $C_{\Psi}^*$ and $S$ may not be self-adjoint. However, if $\psi_{y}\in \dom(C_{\Psi}),\, \forall\, y\in X$,
 then $C_{\Psi}$ is densely defined, $D_{\Psi}\subseteq C_{\Psi}^*$ and $D_{\Psi}$ is closable. Finally, $D_{\Psi}$ is closed if and only if  $D_{\Psi}= C_{\Psi}^*$. 
Then $S=C_{\Psi}^*C_{\Psi} $ is self-adjoint \cite[Lemmas 5.3 and 5.4]{jpa-bal}.
\smallskip

Next, we say that  a    family $\Phi=\{\phi_{x}\}$  is a \emph{lower semi-frame} if 
it satisfies the lower frame condition, that is, there exists   a constant ${\sf m}>0$ such that
\be
{\sf m}  \norm{}{f}^2 \le  \int_{X}  |\ip{\phi_{x}}{f}| ^2 \, \ud \nu(x) , \;\; \forall \, f \in \Hil.
\label{eq:lowersf}
\en
 Clearly, \eqref{eq:lowersf} implies that the family $\Phi $ is total in $\Hil$.
With these definitions, we obtain a nice duality property   between upper and lower semi-frames .  
\begin{prop.}\label{prop:duality}  
 (i) Let $\Psi = \{\psi_{x}\}$ be an upper semi-frame, with upper frame bound {\sf M} and let $\Phi =\{\phi_{x}\}$ be a total family dual to $\Psi$.
Then $\Phi$ is a lower  semi-frame, with lower frame bound ${\sf M}^{-1}$. 

(ii) Conversely, if $\Phi =\{\phi_{x}\}$ is a lower semi-frame, there exists an upper semi-frame
$\Psi= \{\psi_{x}\} $ dual to $\Phi $,  that is, one has, in the weak sense,
$$
f = \int_X  \ip{\phi_{x}}{f} \, \psi_{x} \, \ud\nu(x)  , \;\; \forall\, f\in \dom(C_{\Phi}).
$$
\end{prop.}
A proof may be found in \cite[Lemma 5.5 and Proposition 5.6]{jpa-bal}. In the same paper (Sections 5.6 and 5.7), we have presented
 concrete examples of a non-regular upper semi-frame (the example from affine coherent states discussed in Section \ref{sec:contuppersemiframes})
 and of a lower semi-frame (from wavelets on the 2-sphere).

\section{\sc Frame and semi-frame equivalence}
\label{sec: frame-equiv}

An interesting notion, developed in \cite{jpa-contframes}, is that of \emph{frame equivalence}. Actually there are several different notions here.
In the sequel, we consider the so-called  rank-$n$ frames, but all the statements below are valid \emph{verbatim} for regular upper semi-frames, since only $S$ is involved, not $S^{-1}$.
\medskip

In the general case of a measure space $(X,\nu)$ \cite[Section 2]{jpa-contframes}, a rank-$n$ frame  consists   of a collection of
$n$-dimensional subspaces, one for each $x\in X$, with orthonormal basis\footnote{A general basis would suffice.}  $\{{\psi}^{i}_{x}\}, i = 1, 2, \ldots , n < \infty$, 
and for which there exists a positive operator $S \in GL(\Hil)$, the frame operator,  such that, with weak convergence,\footnote{
Here we use the tensor product notation $(\xi \otimes \ov{\eta })f = \ip{\eta}{f}\xi $ ; in Dirac notation, this means $\xi \otimes \ov{\eta} =| \xi\rangle \langle \eta|. $}
$$ 
\sum_{i=1}^{n}\int_{X} ({\psi}^{i}_{x}\otimes    \ov{\psi^{i}_{x}}) \, \ud \nu (x) = S .
 $$
Introduce the rank-$n$ projection operator
$$
 \Lambda(x) = \sum _{i=1}^{n}{\psi}^{i}_{x}\otimes    \ov{\psi^{i}_{x}}, \quad\mbox{for each } x \in X,
 $$
\vspace*{-2mm}
 so that
$$ 
\int_{X}\Lambda(x) \, d\nu (x) = S,
 $$
or, equivalently,  there exist constants  ${\sf m}>0$ and ${\sf M}<\infty$ such that
\be\label{eq:rankn-frame}
{\sf m}  \norm{}{f}^2 \leq \ip{f}{Sf} = \int_{X}  \norm{}{\Lambda(x)f} ^2 \, \ud \nu(x) \leq {\sf M}  \norm{}{f}^2 ,  \forall \, f \in \Hil.
\vspace*{-2mm}
\en
We denote such a rank-$n$ frame as $ \F (\psi^{i}_{x} , S) $. If we consider only the operator $\Lambda(x)$ or, equivalently, 
the $n$-dimensional subspace spanned by the vectors $\{{\psi}^{i}_{x}\}, i = 1, 2, \ldots , n < \infty$, we will speak of the \emph{reproducing triple}
 $\{\Hil, \Lambda, S\}$.
More generally, we can still speak of a reproducing triple if the rank of $\Lambda(x)$ depends on $x$ (it could even be infinite).
The same definitions apply to an upper semi-frame, where $S$ is bounded and invertible, but the inverse $S^{-1}$ is unbounded.

According to the general theory,  the   frame  (or the regular upper semi-frame) $\F (\psi^{i}_{x} , S) $   admits the reproducing kernel $K^{\Psi}$, \emph{i.e.}, 
the $n\times n$ matrix-valued function $K^{\Psi}$ on $X\times Y$, whose elements are given by
$$
 K^{\Psi}_{ij}(x,y) = \ip{ \psi^{i}_{x} }{S^{-1}\psi^{j}_{y} },        \; x,y \in X, \quad i,j = 1,2, \ldots , n      .     
$$
  $K^{\Psi}$ is called the  \emph {frame kernel}. 
\medskip

  \noi \emph{(i) Similar and unitarily equivalent frames:}
\\[1mm]
Two rank-$n$ frames $\F ( \psi^{i}_{x} , S) $ and $\F( \widetilde\psi^{i}_{x} , \widetilde  S) $ are said to be \emph{similar} if there exists an operator 
$T\in GL(\Hil)$ such that
$$ 
\widetilde\psi^{i}_{x} = T\psi^{i}_{x},  \quad  \widetilde S = TST^{*}.
$$
The two frames are called   \emph{unitarily equivalent} if the operator $T$ is unitary. In both cases, the frame kernel is invariant: $ \widetilde{K}(x,y) = K(x,y) $.
\medskip

  \noi \emph{(ii) Gauge equivalent frames:}
\\[1mm]
Clearly, if  $\{\widetilde \psi^{i}_{x} , \; i=1,2, \ldots ,n\}$ is another orthonormal basis of {\sf Span}$\,\Lambda(x)$, 
 there exists an $n\times n$ unitary matrix $\mathcal U (x)$ such that
$$
\widetilde\psi^{i}_{x} = \sum_{j=1}{\mathcal U}_{ij}(x)\psi^{j}_{x} ,    \quad  i=1,2, \ldots ,n . 
\vspace*{-2mm}
$$
 Then the  frame  $\F (\widetilde\psi^{i}_{x} , \widetilde  S) $ is called  \emph{gauge equivalent}  to $\F( \psi^{i}_{x} , \widetilde  S) $. 
The corresponding kernels  are gauge related:
$$ 
K^{\widetilde \Psi}_{ij}(x,y)  = \ip{\widetilde \psi^{i}_{x} }{S^{-1}\widetilde \psi^{j}_{y} }
  = \sum_{k,l=1}^{n} \overline{{\mathcal U}_{ik}(x)}\,K^{\Psi}_{kl}(x,y)\, {\mathcal U}_{jl}(y), 
 $$
 which we shall also write as 
$ K^{\widetilde \Psi }(x,y) = {\mathcal U (x)}^{*}\,K^{\Psi}(x,y)\,\mathcal U (y)$.
\medskip

  \noi \emph{(iii) Kernel equivalent frames:}
\\[1mm]
Two frames are called  \emph {kernel equivalent } if they are both similar and gauge equivalent :
$$   
 \widetilde\psi^{i}_{x} = \sum_{j=1}{\mathcal U}_{ij}(x)T\psi^{j}_{x} ,    \qquad  i=1,2, \ldots ,n \quad \mbox{ and} \quad\widetilde S = TST^{*}. 
 \vspace*{-2mm} $$
Thus here too the kernels are gauge equivalent.
\medskip

  \noi \emph{(iv) Bundle equivalent frames:}
\\[1mm]
Finally, given any two rank-$n$ reproducing triples $\{\Hil, \Lambda, S\}$ and $\{ \Hil , \widetilde \Lambda , \widetilde S \}$, there exists
 a family of rank-$n$ operators $T(x), x \in X$ on $\Hil$ for which
$$ 
\widetilde \Lambda (x) = T(x)\Lambda(x)T(x)^{*}, \quad  \widetilde \psi^{i}_{x} = \sum_{j=1}^{n}{\mathcal U}_{ij}(x)T(x)  \widetilde \psi^{j}_{x}
 \vspace*{-1mm}$$
In that case, the two frames are called  \emph {bundle equivalent}. Notice that here there is no connection between $S$ and $\widetilde  S$. 
Also two bundle equivalent frames are kernel equivalent if the rank-$n$ operator $T(x), x \in X$ is constant, in the sense that  $ T(x) = T \, \Lambda(x)$.

 To summarize, if we denote by $b\,\sim,\, k\,\sim,\, u\,\sim\,$ and $g\,\sim$ the relations of bundle, kernel, unitary and gauge equivalence,
respectively, between two frames in $\Hil$, then the following hierarchical structure emerges:
\begin{center}
\begin{tabular}{ccccc}
$u\,\sim$ &&&&\\
&$\searrow$&&&\\
&&$ k\,\sim$ &$\rightarrow$  &$b\,\sim$\\
&$\nearrow$&&&\\
$g\,\sim$ &&&&
\end{tabular}
\end{center}

Note that frame equivalence still makes sense for $n=1$, the only difference being that the unitary matrix ${\mathcal U}(x)$ reduces to a phase.
This  shows that the relevant quantity is the operator $\Lambda (x)$ or the subspace it spans. Clearly these considerations bring us directly to the fusion frames     described in Section \ref{subsec:fusionframes}.

In order to make the connection precise, let us assume there is a partition   $X = \bigcup_{j\in J} X_j$,  where $v_j^2 := \int_{X_j} \ud \nu(x) < \infty$
such that $\Lambda (x)$ is constant over each $X_j$, call it $\Lambda (x)\up X_j = \Lambda_{j}$. Note that the rank of 
$ \Lambda_{j}$ could be infinite. Then the relation \eqref{eq:rankn-frame} becomes
$$
{\sf m}  \norm{}{f}^2  \le \sum_{j\in J} v_{j}^2 \,\norm{}{ \Lambda_{j}\,f}^2  \le {\sf M}  \norm{}{f}^2 ,  \forall \, f \in \Hil,
$$
\emph{i.e.},  we obtain precisely a fusion frame. Once again, all this extends immediately to upper semi-frames.

\section*{\sc Acknowledgements}

The authors would like to thank D. Stoeva and O. Christensen for helpful comments and suggestions.  
This work was partly supported by the WWTF project MULAC (`Frame Multipliers: Theory and Application in Acoustics', MA07-025).
The first author acknowledges gratefully the hospitality of  the Acoustics Research Institute, Austrian Academy of Sciences,  Vienna, 
and so does the second author towards the  Institut de Recherche en Math\'ematique et  Physique, Universit\'e catholique de Louvain.

\end{document}